%% file: main.tex
\documentclass[letterpaper]{article} %
\usepackage{aaai2026}  %
\usepackage{times}  %
\usepackage{helvet}  %
\usepackage{courier}  %
\usepackage[hyphens]{url}  %
\usepackage{graphicx} %
\usepackage{natbib}  %
\usepackage{caption} %
\usepackage{algorithm}
\usepackage{algorithmic}
\usepackage{booktabs}
\usepackage{newfloat}
\usepackage{listings}
\DeclareCaptionStyle{ruled}{labelfont=normalfont,labelsep=colon,strut=off} %
\floatstyle{ruled}
\newfloat{listing}{tb}{lst}{}
\floatname{listing}{Listing}
\usepackage{color}

\title{Incremental Data-Driven Policy Synthesis via Game Abstractions}
\author {
    Irmak Sağlam\textsuperscript{\rm 1\thanks{These authors contributed equally.}},
    Mahdi Nazeri\textsuperscript{\rm 1,\rm 2\footnotemark[1]},
    Alessandro Abate\textsuperscript{\rm 2},
    Sadegh Soudjani\textsuperscript{\rm 1,\rm 3},
    Anne-Kathrin Schmuck\textsuperscript{\rm 1}
}
\affiliations {
    \textsuperscript{\rm 1}Max Planck Institute for Software Systems, Kaiserslautern, Germany\\
    \textsuperscript{\rm 2}Department of Computer Science, University of Oxford, Oxford, United Kingdom\\

    \textsuperscript{\rm 3}School of Computer Science, University of Birmingham, United Kingdom\\
    
    \{isaglam, sadegh, akschmuck\}@mpi-sws.org, \{mahdi.nazeri, alessandro.abate\}@cs.ox.ac.uk
}

\newif\ifappendix
\global\appendixfalse

\usepackage{bibentry}
\input{setup/packages.tex}

\input{setup/macros.tex}

\begin{document}

\maketitle

\begin{abstract}
We address the synthesis of control policies for unknown discrete-time stochastic dynamical systems to satisfy temporal logic objectives. We present a data-driven, abstraction-based control framework that integrates online learning with novel incremental game-solving. Under appropriate continuity assumptions, our method abstracts the system dynamics into a finite stochastic (2.5-player) game graph derived from data. Given a
requirement over time 
on this graph, we compute the winning region---i.e., the set of initial states from which the objective is satisfiable---in the resulting game, together with a corresponding control policy.

Our main contribution is the construction of abstractions, winning regions and control policies \emph{incrementally}, as data about the system dynamics accumulates. Concretely, our algorithm refines under- and over-approximations of reachable sets for each state-action pair as new data samples arrive. These refinements induce structural modifications in the game graph abstraction---such as the addition or removal of nodes and edges---which in turn modify the winning region. Crucially, we show that these updates are inherently monotonic: under-approximations only grow, over-approximations only shrink, and the winning region only expands.

We exploit this monotonicity by defining an objective-induced ranking function on the nodes of the abstract game that increases monotonically as new data samples are incorporated. These ranks 
underpin our novel incremental game-solving algorithm, which employs customized gadgets (DAG-like subgames) within a rank-lifting algorithm to efficiently update the winning region. Numerical case studies demonstrate significant computational savings compared to the baseline approach, which re-solves the entire game from scratch whenever new data samples arrive.

\end{abstract}

\begin{links}
    \link{Code}{https://github.com/nazerimahdi/AAAI-26}
\end{links}

\section{Introduction}

\paragraph{Motivation} Guaranteeing correct behavior in safety-critical systems---such as autonomous vehicles, robotic platforms, or air traffic control---requires control policies that satisfy high-level specifications under uncertainty \cite{belta2019formal}. A rich class of such specifications can be expressed using temporal logic \cite{baier2008principles}, which captures complex time-based goals like ``eventually reach a goal while always avoiding obstacles.'' In the presence of stochastic dynamics, where randomness plays a central role, one wants to synthesize a policy that satisfies these specifications with high probability \cite{LSAZ2022automated}. %

\paragraph{Challenges}
A traditional approach for synthesizing policies with temporal logic guarantees is to construct a finite abstraction of the system which encodes the interaction between the controller, the system’s stochastic behavior, and the specification as a game which can be solved by standard techniques \cite{abate2010approximate,majumdar2024symbolic}. However, such abstraction-based controller synthesis methods typically assume known system models, and rely on static abstractions that must be rebuilt from scratch when the knowledge about the system dynamics changes. This makes them poorly suited when system dynamics are unknown and data-driven or online learning settings are used to infer game abstractions which then evolve over time and induce policy updates which need to be handled efficiently.

\paragraph{Contribution}
This paper introduces the first \emph{incremental}  data-driven abstraction-based control framework for discrete-time stochastic systems with unknown dynamics, which refines abstractions and winning regions incrementally---and therefore very efficiently---as new data about the system dynamics becomes available.

\section{Overview}\label{sec:overview}

\input{sections/example}

\section{Learning Abstract Reachable Sets}\label{sec:learning}
\input{sections/learning.tex}

\section{Exploiting Monotonicity}
\input{sections/monotonicity.tex}

\section{Incremental Game Solving}\label{sec:incremental}
\input{sections/incremental}

\section{End-to-End Efficiency}

\input{sections/fixpoint-algorithm.tex}

\section{Experiments}
\input{sections/experiments.tex}

\section{Conclusion}
In this work, we presented an incremental data-driven abstraction-based synthesis algorithm for unknown dynamical systems with additive noise. We first provided a monotone method to construct the abstraction (game graph) from data and proved that the abstraction becomes less conservative as new data become available. Then, we developed a lifting algorithm that enables incremental updates to the solution of the game graph when new data arrives. Finally, we proved that our lifting algorithm can be warm-started from the solution of a fixed-point algorithm for computational efficiency. Numerical results demonstrate the effectiveness of our algorithm, particularly when the game graph is refined locally with newly arriving data.
 
Conceptually, our work presents a novel abstraction-based control technique based on state-space discretizations. It is well known that such discretization-based approaches are hard to scale to large dimensional systems but are able to handle very general non-linear dynamics and reactive specifications. While our approach thereby suffers from scalability limitations, incremental synthesis techniques lend themselves more naturally to compositional frameworks where solutions over subsets of dimensions are iteratively computed to obtain an overall solution. The exploitation of our novel synthesis technique to address scalability concerns in such an iterative compositional framework is an interesting direction for future work.

\section*{Acknowledgments}
This work is funded by the DFG grant 389792660 as part of TRR 248 – CPEC, by the European Union with ERC Auto-CyPheR
grant 101089047 and EIC SymAware grant 101070802, and by the Emmy Noether Grant SCHM 3541/1-1.
Views and opinions expressed are however those of the author(s) only and do not necessarily reflect those of the European Union or the European Research Council Executive Agency. Neither the European Union nor the granting authority can be held responsible for them.
\bibliography{aaai2026}

\newpage               %

\appendix%
\newpage
\input{appendix/related_work}

\input{appendix/app-notation}

\input{appendix/app-sec3}
\input{appendix/app-sec4}
\input{appendix/app-sec5}

\input{appendix/app-sec6}

\input{appendix/app-experiments}
\end{document}

%% file: setup/packages.tex
\usepackage{breakurl}             %
\usepackage{underscore}           %
\usepackage{caption}
\usepackage{tikz}
\usepackage{longtable}
\usepackage{array}
\usepackage{ifthen}
\usepackage[nointegrals]{wasysym}   
\usepackage{lipsum}  
\usepackage{booktabs}

\usepackage{amsthm} %

\newtheorem{lemma}{Lemma}
\newtheorem{definition}{Definition}
\theoremstyle{remark}
\newtheorem*{rem}{Remark}

\usepackage{tikz}
\newcommand*\circled[1]{\tikz[baseline=(char.base)]{
            \node[shape=circle,draw,inner sep=1pt] (char) {#1};}}
            
\usepackage{amssymb}
\usepackage{thmtools}

\usetikzlibrary{trees, arrows,shapes,automata, calc, chains, matrix, positioning, scopes}
 \usetikzlibrary{trees}
 \usetikzlibrary{shapes}
 \usetikzlibrary{fit}
 \usetikzlibrary{shadows}
 \usetikzlibrary{backgrounds}
 \usetikzlibrary{arrows,automata,calc}

\tikzset{
   n/.style= {circle,fill,inner sep=1.5pt,node distance=2cm}
  ,acc/.style={circle,draw,inner sep=3pt,node distance=2cm}
  ,phantom/.style={circle},
  ,arr/.style={->, >=stealth, semithick, shorten <= 3pt, shorten >= 3pt}
}
\usetikzlibrary{patterns,snakes}

\usepackage{mathtools,amsfonts}

\usepackage{thm-restate}
\usepackage{paralist}
\usepackage{xcolor}
\usepackage{soul}
\usepackage[utf8]{inputenc}
\usepackage{xspace}
\usepackage{stmaryrd} %
\usepackage{enumitem}
\usepackage{bm}

\usepackage[capitalise]{cleveref}
\crefname{definition}{Def.}{Def.}
\crefname{clm}{Claim.}{Claim.}
\crefname{problem}{Prob.}{Prob.}
\crefname{algorithm}{Algo.}{Algo.}
\crefname{theorem}{Thm.}{Thm.}
\crefname{lemma}{Lem.}{Lem.}
\crefname{appendix}{App.}{App.}
\crefname{line}{line}{line}
\crefname{section}{Sec.}{Sec.}
\crefname{corollary}{Cor.}{Cor.}
\crefname{example}{Ex.}{Ex.}
\crefname{line}{Line}{Lines}
\crefname{item}{Item}{Items}
\crefname{subsection}{Subsec.}{Subsec.}
\crefname{inpara}{Inline Item}{Inline Item}
\crefname{figure}{Fig.}{Fig.}
\crefname{table}{Table}{Tables}

%% file: setup/macros.tex
\newcommand{\Cpre}{\mathsf{Cpre}^1}
\newcommand{\Cpresys}{\mathsf{Cpre}^0}
\newcommand{\Lpre}{\mathsf{Lpre}^{\forall}}
\newcommand{\Lpreexists}{\mathsf{Lpre}^{\exists}}
\newcommand{\Preexists}{\mathsf{Pre}_1^{\exists}}
\newcommand{\Preforall}{\mathsf{Pre}_0^{\forall}}
\newcommand{\subgraph}{\mathsf{H}}
\newcommand{\subgr}{\mathsf{F}}

\newcommand{\Prog}{\mathtt{pr}}
\newcommand{\Progc}{\mathtt{pr}^{\circ}}

\newboolean{showcomments}
\setboolean{showcomments}{true} %

\ifthenelse{\boolean{showcomments}}
{ %
  \newcommand{\IS}[1]{\textcolor{red!50}{[\textbf{IS:} #1]}}
  \newcommand{\AKS}[1]{\textcolor{violet!50}{[\textbf{AKS:} #1]}}
  \newcommand{\MN}[1]{\textcolor{cyan!50}
  {[\textbf{MN:} #1]}}
  \newcommand{\SaS}[1]{\textcolor{purple!50}{[\textbf{SS:} #1]}}
}
{ %
  \newcommand{\IS}[1]{}
  \newcommand{\AKS}[1]{}
  \newcommand{\MN}[1]{}
  \newcommand{\SaS}[1]{}
}

\newcommand{\Wsys}{\mathsf{Win}_0}
\newcommand{\Wenv}{\mathsf{Win}_1}
\newcommand{\Win}{\mathsf{Win}}

\newcommand{\XX}{\mathbf{B}}

\newcommand{\gamegraph}{G}
\newcommand{\ltup}[1]{\langle #1\rangle}
\newcommand{\spec}{\Phi}

\newcommand{\Inf}{\mathsf{Inf}}

\newcommand{\cupdot}{\mathbin{\mathaccent\cdot\cup}}
\newcommand{\Psys}{Player 0\xspace}
\newcommand{\Penv}{Player 1\xspace}

\newcommand{\Vsys}{V_0}
\newcommand{\Venv}{V_1}

\newcommand{\tup}[1]{\left( #1\right)}
\definecolor{colored}{rgb}{1,0.5,0}

\usepackage{pst-plot} %
 \usepackage{etoolbox} %
\usepackage{tcolorbox} %

\newtcolorbox{resp}[1][]{%
	enhanced jigsaw,%
	colback=gray!5!white,%
	colframe=gray!80!black,%
	size=small,%
	boxrule=1pt,%
	halign title=flush center,%
	coltitle=black,%
	breakable,%
	drop shadow=black!50!white,%
	attach boxed title to top left={xshift=1cm,yshift=-\tcboxedtitleheight/2,yshifttext=-\tcboxedtitleheight/2},%
	minipage boxed title=3cm,%
	boxed title style={%
		colback=white,%
		size=fbox,%
		boxrule=1pt,%
		boxsep=2pt,%
		underlay={%
			\coordinate (dotA) at ($(interior.west) + (-0.5pt,0)$);
			\coordinate (dotB) at ($(interior.east) + (0.5pt,0)$);
			\begin{scope}[gray!80!black]
				\fill (dotA) circle (2pt);
				\fill (dotB) circle (2pt);
			\end{scope}
		}%
	},%
	#1%
}

%% file: sections/example.tex
\begin{figure*}[t]
  \centering
  \includegraphics[width=\textwidth]{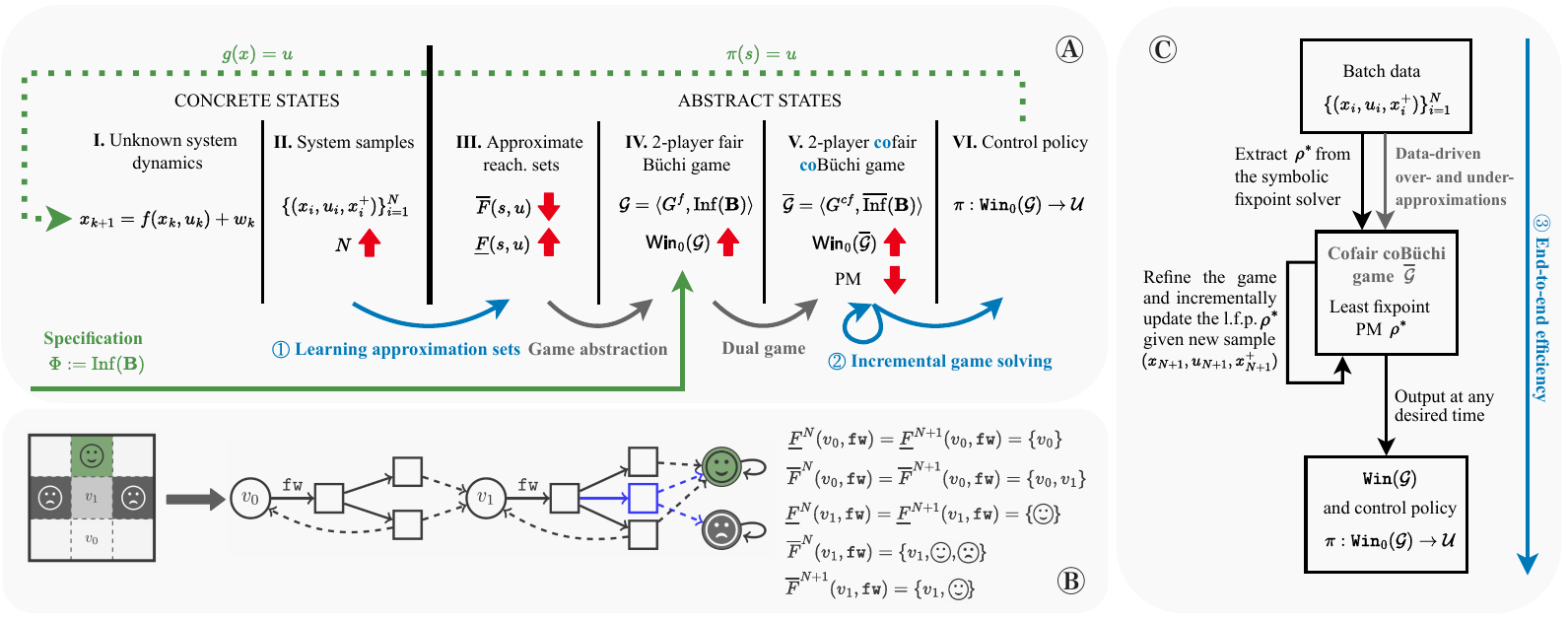}
  \vspace{-0.5cm}
  \caption{Overview of (A) incremental synthesis, (B) motivating example, and (C) end-to-end data-driven synthesis. Our contributions \circled{1},\circled{2}, and \circled{3}, are highlighted in blue, along with the proven learning-induced monotonicity highlighted in red.
  }
  \vspace{-0.2cm}
  \label{fig:overview}
\end{figure*}

The overview of our new incremental synthesis framework is depicted in Fig.~\ref{fig:overview} (A) \& (C) along with a running example in Fig.~\ref{fig:overview} (B) to concretize our contribution.

\paragraph{Running Example}
Consider a simple 2D car which is only moving in one direction (forward). The car dynamics can be modeled by a \emph{stochastic differential equation} (Fig.~\ref{fig:overview} (A)-I). However, we have no access to this model and can only collect a large set of (noisy) samples $(x,u,x^+)$ (Fig.~\ref{fig:overview} (A)-II). That is, we can put the car in different positions $x$, choose a (constant) control input $u$ and observe its next position $x^+$ after a fixed time period $\tau$.

In addition, we are given a \emph{temporal specification} $\spec$ (Fig.~\ref{fig:overview} green) which requires (among other obligations) that the car should reach a target location eventually, after passing through a gate. In order to formalize this specification, the workspace is divided into regions of interest, as examplified in Fig.~\ref{fig:overview} (B)-left. Here, the black cells represent walls, the green cell represents the target, and the gray cell represents the gate. The car is initially in the cell $v_0$.  

\paragraph{Problem Statement}
Consider the dynamics of the car generally modeled as a discrete-time, continuous-state stochastic system $\mathcal{S}$ such that at time step $k$,
\begin{equation}
\label{eq:dynamics}
    x_{k+1} = f(x_k, u_k) + w_k,
\end{equation}
with state $x_k\in \mathcal{X}$, control input $u_k\in \mathcal{U}$, and noise $w_k\in \Omega$. %
The function $f$ denotes the nominal (noiseless) dynamics.

The \emph{overall problem} we want to solve, is to compute a control policy $g: \mathcal{X}\rightarrow \mathcal{U}$ which maps the current state $x_k$ of the car to an input $u_k$ s.t.\ the resulting \emph{closed loop} trajectory $\xi=x_0x_1x_2\dots$ eventually visits the target cell after passing through the gate with probability one, i.e.\ \emph{almost surely} satisfying $\spec$. The challenge in solving this problem is to gather enough knowledge about the system dynamics and represent this knowledge in an appropriate way, %
such that a control policy can be computed which always enforces such strong correctness guarantees when used to control the actual system, e.g., the car.

\paragraph{Existing Approaches}
To address this problem, a discretization of the state space (e.g.\ by a grid as depicted in~\cref{fig:overview} (B)-left) is typically utilized. Each cell is an \emph{abstract state}, and transition probabilities between these cells can either be computed using the model (if it is known) or learned from samples. This gathered knowledge can then be used to build different types of abstractions (see \citet{LSAZ2022automated} for a recent review; seminal works include \citet{tabuada2009verification, abate2008probabilistic,Soudjani2013adaptive} with recent extensions to data-driven approaches \citet{gracia2023distributionally,nazeri2025data, DBLP:journals/corr/abs-2508-15543}).

These abstractions are formulated as a \emph{game} between the controller (Player 0) and the environment (Player 1), as illustrated in Fig.~\ref{fig:overview} (B), with controller (circle) and environment (box) vertices capturing their interplay.
Intuitively, due to the state-space discretization, the incomplete system knowledge, and the stochastic nature of the system (noise), multiple cells are reachable by the car from a particular cell with the same control input. The environment player captures this non-determinism via the transition structure of the game. 

Historically, such \emph{two-player games on finite graphs} are used to abstract the computation flow of reactive programs in computer science and many algorithmic solution procedures exist \cite{gamesongraphs}. 
Given a temporal logic specification $\spec$ (e.g., visiting the target after passing the gate) over a two-player game (e.g., resulting from abstraction) a solution of the game constitutes a control policy $\pi$ that continuously reacts to the actions of the environment s.t.\ $\spec$ holds. Such specifications are widely used to formalize both safety and strategic liveness objectives for cyber-physical systems, e.g.\ for mobile robot navigation \cite{KressGazitFainekosPappas2009,kress2018synthesisForRobotsReview} or in autonomous driving~\cite{AlthoffBeltaReviewFMCEforAutonomousDriving}.

If a control player state allows for a control policy that fulfills the specification, it is called \emph{winning} and is contained in the \emph{winning region} $\Wsys$. Algorithmic games solving techniques typically compute the winning regions first (Fig.~\ref{fig:overview} (A)-IV) and then derive a control policy (Fig.~\ref{fig:overview} (A)-VI). When game abstractions are built carefully, these control policies $\pi$ can be refined to a controller $g$ which controls the underlying dynamical system s.t.\ the above problem statement is successfully addressed (Fig.~\ref{fig:overview} green dashed arrow).

\paragraph{Limitation}
The major limitation of the above approach is its non-flexibility when new knowledge about the system dynamics is obtained. Given either the system dynamics or a \emph{fixed set} of samples, a game abstraction is computed and solved. When new information about the system dynamics is obtained, the entire abstraction needs to be recomputed and previous knowledge about the game and its solution (i.e., the winning region) are lost. %

\paragraph{Contribution}
To address this limitation, we provide three distinct contributions:
\begin{enumerate}
    \item[\circled{1}] We introduce a novel learning-based abstraction technique which learns under- and over-approximations of reachable sets. This enables incremental refinements of game-abstractions (indicated in Fig.~\ref{fig:overview} by red up/down arrows) when new data samples arrive. 
    \item[\circled{2}] We derive a novel game solving algorithm that incrementally updates the winning region by exploiting the monotonicity of game graph updates. %
\end{enumerate}
While \circled{1} can be combined with the abstraction-based controller synthesis tool \texttt{FairSyn} \cite{majumdar2024symbolic,majumdar2023flexible} when system dynamics are unknown, \circled{1} can also be combined with \circled{2} to obtain a novel efficient incremental synthesis algorithm when a small set of new samples arrives. In practice, however, we need both. We typically first have a large data set and compute an initial game abstraction along with its winning region, which we would then like to refine if (a small number of) new samples arrive. Unfortunately, initialization of \circled{2} with the solution of %
a symbolic fixpoint solver is highly non-trivial. Our final contribution is:
\begin{enumerate}
    \item[\circled{3}] We provide an efficient end-to-end algorithm that initializes the incremental data-driven abstraction based controller synthesis algorithm which combines \circled{1} and \circled{2} with the \emph{batch} solution which combines \circled{1} and a symbolic fixpoint solver %
    (see Fig.~\ref{fig:overview} (C) for an illustration).
\end{enumerate}
On a conceptual level, all above contributions are enabled by a formal connection of 
the monotonicity property of abstraction learning and the monotonicity requirements of incremental game solving (highlighted by red arrows in~\cref{fig:overview}). 

\smallskip
Before we formalize our contributions, we overview them briefly using the introduced running example.

\paragraph{\circled{1} Learning Approximate Reachable Sets}
Going back to our running example, let $s$ be an abstract state (i.e., a grid cell) and $\texttt{fw}$ the action \emph{forward}. Then the under-approximation $\underline{F}(s, \texttt{fw})$ carries the information \emph{if the car takes the action $\texttt{fw}$ from $s$ infinitely often, which cells $s'$ will it end up \textbf{for sure} (i.e., with probability one)?} On the other hand the over-approximation  $\overline{F}(s, \texttt{fw})$ carries the information \emph{if the action $\texttt{fw}$ is taken from $s$ infinitely often, which cells \textbf{can it} end up in (i.e., with positive probability)?}
For our concrete car example, it might happen that after applying $\texttt{fw}$ in $v_0$, the car ends up in $v_0$ in the next time step (i.e., $v_0\in\overline{F}(v_0,\texttt{fw}$), while we know that \emph{eventually} $v_1$ will be reached under this action, (i.e., $v_1\in\underline{F}(v_0,\texttt{fw})$).

In this paper, we present a novel algorithm to compute $\overline{F}^N, \underline{F}^N$ from a data set $\mathbf{D_N} = \{(x_i, u_i, x_i^+)\}_{i=1}^{N}$ for each abstract state-input pair when the specification $\spec$ is Büchi. Given these under- and over-approximations, \citet{majumdar2024symbolic} formalizes how an abstract fair Büchi game $\langle G^f,\spec\rangle$ can be built whose solution solves the above problem statement (\emph{game abstraction} \cref{fig:overview} (A) III$\to$IV). Intuitively, such games (an example is~\cref{fig:overview} (B)) include \emph{fair environment edges} (dashed) which restrict the choices of the environment player such that such edges need to be taken infinitely often, when the source vertex is seen infinitely often, which abstractly captures the semantics of $\underline{F}$ as exemplified above. 

\paragraph{\circled{2} Incremental Game Solving}
It is well known that solving graph games incrementally is difficult. The very few existing techniques \cite{ChatterjeeHenzinger14, SaglamST24} rely on a ranking function, called \emph{progress measure} (PM). %
The advantage of PM algorithms is that PMs are (1) only locally updated and (2) allow to extract the winning region once their value has stabilized. %
Therefore, if the mentioned local PM updates are efficiently implementable and the game graph changes are mild, PM algorithms might be used to efficiently re-compute the winning regions.
Unfortunately, the PM of a game depends on the specification and they are often either not known or too complicated to compute efficiently. Further, to obtain an efficient PM-based incremental algorithm for a game class, 
we should re-solve the game strictly under monotonic graph modifications, i.e., modifications that can only increase the PM.

In this work, our main contribution is to 
(i) show that our online learning framework yields monotonic graph modifications for the \emph{dual of} the abstract fair Büchi game $\mathcal{G}=\langle G^f,\spec\rangle$---i.e.,  the induced \emph{co}fair \emph{co}Büchi game $\overline{\mathcal{G}}=\langle G^{cf},\overline{\spec}\rangle$, and
(ii) construct a PM for \emph{co}fair \emph{co}Büchi games (Fig.~\ref{fig:overview} (A)-V), which allows to incrementally refine the winning region $\Wsys(\mathcal{G})$. 

While the illustrated game in Fig.~\ref{fig:overview} (B) is built as a fair Büchi game, due to their duality, it can simply be re-interpreted as a \emph{co}fair \emph{co}Büchi game by swapping the players (now fair moves are owned by the controller) and the winning condition (now $\smiley$ should only be visited finitely often \emph{by the (new) control player}). 
In the illustrated fair Büchi game we see that a refinement of $\overline{F}(v_1, \texttt{fw}) = \{v_1, \smiley, \frownie\}$ to $\{v_1, \smiley\}$ leads to the removal of the purple edges in the game. In the modified game (only black edges) we see that $\frownie$ being unreachable from $v_0$ lets the control player win the game. This is due to the fair (dashed) edges: Whenever an environment node is passed infinitely often, all its dashed edges are taken infinitely often as well. Therefore, every play starting from $v_0$ eventually reaches $\smiley$. 
Intuitively, the refinements of approximation sets ($\underline{F}$ expanding/$\overline{F}$ shrinking) yield graph modifications that only expand the controller winning region in the fair Büchi game. This means, the fair Büchi PM can only \emph{decrease} after graph modifications. In turn, the cofair coBüchi PM in the dual cofair coBüchi game can only \emph{increase}. This makes the learning-induced graph modifications \emph{monotonic} on the dual cofair coBüchi game.

\paragraph{\circled{3} End-to-End Efficiency}
We note that constructing a PM for cofair coBüchi games is a novel technique for algorithmic game solving. It uses a novel gadget (a small DAG-like subgame) construction to turn a cofair coBüchi game into a regular coBüchi game. This allows to \emph{pull} the (known) coBüchi PM on the reduced coBüchi game (introduced by~\citet{Jurdzinskismall00}) \emph{back} to the cofair coBüchi game, which yields a PM for cofair coBüchi games. The gadgets introduced in this work are inspired by the gadgets in~\citet{HausmannPSS24} but are more compact. However, a crucial feature of our novel gadgets is their consistency with the (known) cofair coBüchi fixpoint algorithm~\cite{banerjee2023fast}. This allows us to use
a symbolic fixpoint solver to initialize our algorithm.\footnote{The initialization step can be carried out by any off‐the‐shelf symbolic fixpoint solver (e.g.\ \texttt{FairSyn}).  We instead use our local solver because interfacing with \texttt{FairSyn} is nontrivial and it provides no specialized optimizations for the two‐nested fixpoints arising in Büchi/coBüchi games.}

%% file: sections/learning.tex
This section details our first contribution (cf. Fig.~\ref{fig:overview}, \circled{1}). 

\subsection{Learning Bounds on System Dynamics}
Closely related to prior work \cite{zabinsky2003optimal, beliakov2006interpolation, jin2020data}, we first present a data-driven method for learning upper and lower bounds of the unknown dynamics $f$. %

\paragraph{Assumptions}
Our approach relies on two assumptions:\\
\begin{inparaitem}[$\blacktriangleright$]
\item \textbf{(Lipschitz continuity)} The system dynamics $f$ are unknown but Lipschitz continuous in $x$, with a known upper bound $L_X$ on its Lipschitz constant. Formally,
$$
\forall x, y \in \mathcal{X},\; \forall u \in \mathcal{U}: \quad \|f(x, u) - f(y, u)\| \leq L_X \cdot \|x - y\|.\footnote{We use $||v||$ to denote the $L^\infty$ norm of vector $v$.}
$$
\item \textbf{(Bounded noise support)} The noise $w_k$ is supported on a known compact set
$\Omega = [l, h] \subset \mathbb{R}^n$, although the underlying probability distribution $\mathbb{P}$ is unknown.\footnote{
For vectors $l, h \in \mathbb{R}^n$, $[l, h]$ is the axis-aligned hyperrectangle defined by the Cartesian product
$
[l, h] := \prod_{i=1}^n [l(i), h(i)]$. }
\end{inparaitem}

These assumptions imply that our framework only requires state measurements without requiring the noise measurements. It relies only on the knowledge of the support of the noise without requiring any information about its probability distribution or access to independent noise samples.
This allows for practical settings in which observations are easily obtainable, for example, by recording the system in either open- or closed-loop operation. We note that using conservative values for both the Lipschitz constants and the noise support preserves the soundness of our approach.

\paragraph{Learning bounds from Samples}
Given a dataset $\mathbf{D_N}$ and Lipschitz constant $L_X$, we fix an arbitrary state $x_* \in \mathbb{R}^n$ and control input $u_* \in \mathcal{U}$. This induces the subset of data samples 
$
\mathbf{D_N}(u_*) = \{(x_i, u_i, x_i^+) \in \mathbf{D_N} \mid u_i = u_*\}\subseteq \mathbf{D_N}.
$ s.t.\
for all $(x_i, u_*, x_i^+) \in \mathbf{D_N}(u_*)$, 
$
x_i^+ = f(x_i, u_*) + w_i.
$
Since the noise term $w_i$ is unknown but bounded within the compact set $\Omega = [l,h]$, we obtain the bounds (see Fig.~\ref{fig:learning})
\begin{align}\label{eq:bounds}
f(x_*, u_*) \in~ &[\check{f}(x_*, u_*|\mathbf{D_N}), \hat{f}(x_*, u_*|\mathbf{D_N})],~\text{s.t.}\\
\check{f}(x_*, u_*|\mathbf{D_N}) := &\max_{\substack{(x_i, u_*, x_i^+) \\ \in \mathbf{D_N}(u_*)}}
\left(x_i^+ - L_X\cdot\|x_i - x_*\|\right) - h,\notag\\
\hat{f}(x_*, u_*|\mathbf{D_N}) := &\min_{\substack{(x_i, u_*, x_i^+) \\ \in \mathbf{D_N}(u_*)}}
\left(x_i^+ + L_X\cdot\|x_i - x_*\|\right) - l.\notag
\end{align}

\begin{figure}
    \centering
    \includegraphics[width=0.8\linewidth]{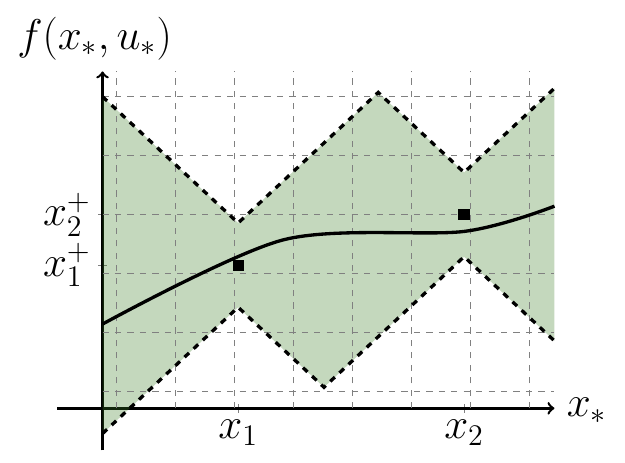}
    \caption{Unknown function $f(x_*, u_*)$ (black), two noisy observations $x_1$, $x_2$ (small squares) and learned lower and upper bounds $\check{f}(x_*, u_*|\mathbf{D_2})$, $\hat{f}(x_*, u_*|\mathbf{D_2})$ (dashed).}
    \label{fig:learning}
\end{figure}
\subsection{Constructing Approximate Reachable Sets}
As explained in the overview section, this paper follows an abstraction-based synthesis paradigm which requires the formal definition of the abstract state and action space (i.e., going from the left to the right side of Fig.~\ref{fig:overview} (A)). 

\paragraph{Abstract States and Transitions} 
We partition the continuous state space $\mathcal{X} \subseteq \mathbb{R}^n$ into $v$ disjoint regions $\{R_i\}_{i=1}^v$ such that $\mathcal{X} = \bigcup_{i=1}^v R_i$ and $ R_i \cap R_j = \emptyset$ for all $i \neq j$.
Each region $R_i$ is represented by an \emph{abstract state} $s_i\in S := \{s_i\}_{i=1}^{v}$. We define a translation function s.t.\ for all $1 \leq i \leq v$:
$\mathcal{T}(x) = s_i$ for all $x \in R_i$ and
$\mathcal{T}^{-1}(s_i) = R_i$.
We define 
$
T(s \mid x, u) = \mathbb{P}_w\big(f(x, u) + w \in \mathcal{T}^{-1}(s)\big)
$
as the probability of reaching $s$ from $x$ under input $u$.

\paragraph{Approximate Reachable Sets}
We now derive under- and over-approximations of abstract reachable states from the learned bounds via Eq.~\ref{eq:bounds} (cf. Fig.~\ref{fig:overview} (A)-III). 
\begin{restatable}{thm}{ouapproximation}\label{thm:ouapproximation}
The following sets of abstract states,
\begin{align*}
\underline{F}(s, u|\mathbf{D_N}) &:= \{s' \in S \mid \mathcal{T}^{-1}(s') \subseteq \underline{f}(s, u|\mathbf{D_N})\} \text{ and}\\
\overline{F}(s, u|\mathbf{D_N}) &:= \{s' \in S \mid \mathcal{T}^{-1}(s') \cap \overline{f}(s, u|\mathbf{D_N}) \neq \emptyset\},
\end{align*}
define\footnote{We sometimes use $\overline{F}^N(s, u)$ to denote $\overline{F}(s, u|\mathbf{D_N})$, and often also omit the dependence on $\mathbf{D_N}$ and use $\overline{F}(s, u)$. Similarly for $\underline{F}$.}
an under- and over-approximation of the forward reachable set of $s$ under input $u$, respectively,
where 
\begin{align*}
&\overline{f}(s, u|\mathbf{D_N}) :=\\
&\left[\min_{x_* \in \mathcal{T}^{-1}(s)} \check{f}(x_*, u|\mathbf{D_N}) + l,\!\max_{x_* \in \mathcal{T}^{-1}(s)} \hat{f}(x_*, u|\mathbf{D_N}) + h\right]\!,\\
& \text{and }\underline{f}(s, u|\mathbf{D_N}):=\\
&\left[\max_{x_* \in \mathcal{T}^{-1}(s)} \hat{f}(x_*, u|\mathbf{D_N}) + l,\!\min_{x_* \in \mathcal{T}^{-1}(s)} \check{f}(x_*, u|\mathbf{D_N}) + h\right]\!.
\end{align*}
\end{restatable}

\subsection{Abstract Fair Büchi Games}
\input{sections/game-abstraction}

%% file: sections/game-abstraction.tex
To complete the picture, we recall from \cite{majumdar2024symbolic} how $\underline{F}(s, u)$ and $\overline{F}(s, u)$ are used to construct a fair game graph (cf. Fig.~\ref{fig:overview} (A) III$\to$IV).  

\paragraph{(Co)Fair Game Graphs}
Formally, a \emph{fair game graph} is a tuple $\gamegraph^f= (\gamegraph, E^f)$ where 
 $\gamegraph= \tup{V,\Vsys,\Venv,E}$ is a game graph s.t.\ $(V,E)$ is a finite directed graph with \emph{edges} $ E $ and \emph{vertices} $ V $ partitioned into \Psys and \Penv vertices, $\Vsys$ (circles) and $\Venv$ (squares), respectively. \emph{Fair edges} (dashed) $E^f \subseteq E \cap (\Venv \times V)$ is a subset of edges originating from $\Venv$ vertices and $V^f = \{ v \in V \mid E^f(v) \neq \emptyset\}$ are called \emph{fair vertices}.
A game graph is called \emph{normal} (or \emph{non-fair}) if $E^f = \emptyset$ and is shown by $G$. It is called \textbf{cofair} if $V^f \subseteq \Vsys$ instead of $\Venv$, and is shown by $G^{cf}$.
 
A play $\xi = v_1 v_2 \ldots$ is an infinite sequence of successive vertices on the game graph and is called \emph{fair} if for each $v \in V^f$, $v \in \Inf(\xi) \Rightarrow E^f(v) \subseteq \Inf(\xi)$. Here, $\Inf(\xi)$ is the set of vertices visited infinitely often in $\xi$.

\paragraph{Games} A game is a tuple $\ltup{G^{\circ}, \spec}$ of a (normal/fair/cofair) game graph $G^\circ$ with a specification $\spec$, which is a set of \Psys-winning plays. A play $\xi$ is winning for \Psys (i) in a normal game, iff $\xi\in\spec$, (ii) in a fair game iff $\xi \in \spec $ or $\xi$ is \emph{not fair}, (iii) in a cofair game iff $\xi \in \spec $ and $\xi$ is \emph{ fair}. Otherwise, $\xi$ is winning for \Penv. A Player $j \in \{0, 1\}$ strategy is a function $\sigma: V^* \cdot V_j \to V$ where $(v, \sigma(h\cdot v))\in E$ for every $h \in V^*$, and $\xi = v_1 v_2 \ldots$ is a $\sigma$-play iff for every $i$, $v_i \in V_j \Rightarrow \sigma(v_1 \ldots v_i) = v_{i+1}$. 
A Player $j$ strategy $\sigma$ is \emph{winning} from $v_1$ iff every $\sigma$-play $\xi = v_1 v_2 \ldots$ is Player $j$-winning. A node $v$ is in the winning region of Player $j$ ($\texttt{Win}_j(G, \spec)$) iff there is a Player $j$ strategy winning from $v$. For all the games we consider, $V = \Wsys(G, \spec) \uplus \Wenv(G, \spec)$.

\paragraph{Abstract fair game graph}
Let $\underline{F}(s, u)$ and $\overline{F}(s, u)$ be the learned approximate reachable sets for every abstract state and input pair $(s,u)$, as in Thm.~\ref{thm:ouapproximation}. 
They induce the following abstract fair game graph $\gamegraph^f$ \cite{majumdar2024symbolic}: 
\begin{align*}
    \Vsys = \,&\{s \mid s \in S\}, \\
\Venv = \,&\{ s^u \mid s \in \Vsys \text{ and } u \in \mathcal{U}\} \,\, \cup \\
& \textstyle \bigcup_{i \in [0, m^{s, u}]} s^u_i \text{  where } m^{s, u} = | \overline{F}(s, u) | - | \underline{F}(s, u) |. 
\end{align*}
For every $s \in S, u \in \mathcal{U}$, assign a fixed order to $\overline{F}(s, u) \setminus \underline{F}(s, u) = \{\overline{s}^{s, u}_1, \ldots, \overline{s}^{s, u}_{m_{s,u}}\}$. The edges are then defined as  
\begin{align*}
    E^f = \,&\{ (s^u_i, \underline{s}) \mid i \in [0, m^{s,u}], \, \underline{s} \in \underline{F}(s, u)\}\, \cup \\
    &\{ (s^u_i, \overline{s}^{s,u}_i) \mid i \in [1, m^{s,u}] \}, \\
    E = \,&E^f \,\cup\, \{(s, s^u) \mid u \in \mathcal{U} \}.
\end{align*}
See Fig.~\ref{fig:overview} (B) for an example, and the appendix for more details. Note that in the constructed game each abstract state is represented by a $\Vsys$ node $s$ (circle), and for each input $u\in \mathcal{U}$ (with $\mathcal{U}=\{\texttt{fw}\}$ in the example), a $\Venv$ node $s^u$ (first layer box-state after circle) is reached. Every $s^u$ has $(m^{s,u}+1)$ successors, all of which are fair nodes (second layer of box-states)---with all outgoing edges being fair---and all of them have all the states in $\underline{F}(s, u)$ as their successors. All but one fair node (i.e.,\ $s^u_0$) have exactly one more state from $\overline{F}(s, u) \setminus \underline{F}(s, u)$ as a (fair) successor.

It is proven by \citet{majumdar2024symbolic} that given any specification $\spec$ in Linear Temporal Logic \cite{baier2008principles} over $G^f$, the solution of the resulting game $\ltup{ G^f,\spec}$ provides a control policy $\pi$ that can be refined into a feedback controller $g$. Within this paper, we restrict attention to a subclass, Büchi and coBüchi specifications.

\paragraph{Büchi and coBüchi}
A Büchi specification $\spec = \Inf(\XX)$ is the set of all plays that visit a node from the set $\XX \subseteq V$ infinitely often. A coBüchi specification $\overline{\Inf}(\XX)$ is the set of plays that do not visit any node from $\XX$ infinitely often.

We assume to have an initial Büchi specification $\spec = \Inf(\XX)$ where $\XX \subseteq S$ is a set of abstract states to be seen infinitely often. Then, through the above construction we obtain an abstract (fair Büchi) game $\ltup{G^f, \Inf(\XX)}$ whose solution provides a control policy that satisfies $\Inf(\XX)$ with probability one (i.e. \emph{almost surely}) in the dynamical system.

%% file: sections/monotonicity.tex
This section details our next contribution (cf. Fig.~\ref{fig:overview} 
red arrows), connecting (i) the monotonicity of the learning step to (ii) monotonic game graph modifications and (iii) how to exploit the latter for incremental game solving.

\input{sections/monotonicity-learning.tex}

\paragraph{Game Updates}
Let $\mathcal{G}_N = \langle\gamegraph^f_N, \Inf(\XX)\rangle$ and
$\mathcal{G}_{N+1} = \ltup{\gamegraph^f_{N+1}, \Inf(\XX)}$
be abstract fair Büchi games induced by the under- and over-approximations  $\underline{F}^{N}(s, u)$, $\overline{F}^N(s, u)$ and $\underline{F}^{N+1}(s, u)$, $\overline{F}^{N+1}(s, u)$, for all $s \in S, u \in \mathcal{U}$, respectively. That is, $\gamegraph^f_{N+1}$ is $G^f_N$ modified due to a new data sample. Then the modified game has an \emph{expanded} winning region.

\begin{restatable}{thm}{dualmonotonicitya}\label{thm:dualmonotonicitya}
Let $V_{N+1}$ be the vertex set of $G^f_{N+1}$. Then,
    $\Wsys(\mathcal{G}_{N+1}) \supseteq \Wsys(\mathcal{G}_N) \cap V_{N+1}$. %
\end{restatable}

\paragraph{Monotonicity on the dual game}
As discussed in Sec.~\ref{sec:overview}, our second contribution (cf. Fig.~\ref{fig:overview}, \circled{2}, Sec.~\ref{sec:incremental}) introduces a novel \emph{progress-measure} (PM) based incremental game solving algorithm. These algorithms require \emph{monotonic} graph modifications---the \Psys winning region can only \emph{shrink} after a graph modification. This is the opposite of what we have in~\cref{thm:dualmonotonicitya}. Luckily, we automatically obtain this monotonicity on the \emph{dual} of the abstract fair Büchi game.

Fair Büchi and cofair coBüchi games are dual, in the following sense: A  cofair coBüchi game $\overline{\mathcal{G}} = \ltup{\gamegraph^{cf}, \overline{\Inf}(\XX)}$ is the dual of a fair Büchi game $\mathcal{G} = \ltup{G^f, \Inf(\XX)}$ where $\gamegraph^f$ and $\gamegraph^{cf}$ are identical except for the ownership of the nodes is swapped. Then, $\Wsys(\mathcal{G}) = \Wenv(\overline{\mathcal{G}})$
and vice versa. %
The following result is a simple corollary of Thm.~\ref{thm:dualmonotonicitya}.

\begin{restatable}{corollary}{monotonicity}[\cref{thm:dualmonotonicitya}]\label{cor:monotonicity}
    Let $\mathcal{G}_N$ and $\mathcal{G}_{N+1}$ be defined as in~\cref{thm:dualmonotonicitya}, and let $\overline{\mathcal{G}}_N$ and $\overline{\mathcal{G}}_{N+1}$  be their dual cofair coBüchi games. 
    Then, $\Wsys(\overline{\mathcal{G}}_{N+1}) \subseteq \Wsys(\overline{\mathcal{G}}_{N}) $; that is, the learning-induced graph modifications are monotonic on the dual of the abstract fair Büchi game.
\end{restatable}
Therefore, the incremental game solving algorithm presented in the next section will incrementally compute $\Wsys(\overline{\mathcal{G}})$, $\Wenv(\overline{\mathcal{G}})$ for the dual $\overline{\mathcal{G}}$ of the abstract game $\mathcal{G}$. As 
$\Wenv(\overline{\mathcal{G}}) = \Wsys(\mathcal{G})$,
this automatically induces  $\Wsys(\mathcal{G})$.

%% file: sections/monotonicity-learning.tex
\paragraph{Reachable Sets}
We first establish a monotonic behavior of the bounds $\check{f}$ and $\hat{f}$ under a new data sample. Intuitively, as more samples become available, these bounds capture more precise information about the unknown system dynamics.

\begin{restatable}{thm}{mono}\label{thm:mono}
For all $u_* \in \mathcal{U}$, $x_* \in \mathcal{X}$, and $N \in \mathbb{N}$, the following monotonicity properties hold:
$$\check{f}(x_*, u_*|\mathbf{D_N}) \preccurlyeq \check{f}(x_*, u_*|\mathbf{D_{N+1}}),$$
$$\hat{f}(x_*, u_*|\mathbf{D_N}) \succcurlyeq \hat{f}(x_*, u_*|\mathbf{D_{N+1}}).$$
\end{restatable}
This directly implies that the over- and under-approximations of the abstract reachable sets also become tighter. %

\begin{restatable}{thm}{UAOA}\label{thm:UAOA}
For all $s \in S$, $u \in \mathcal{U}$, and $N \in \mathbb{N}$, the following monotonicity properties hold:
$$\underline{F}(s, u|\mathbf{D_N}) \subseteq \underline{F}(s, u|\mathbf{D_{N+1}}),$$
$$
\quad \overline{F}(s, u|\mathbf{D_{N+1}}) \subseteq \overline{F}(s, u|\mathbf{D_N}).$$
\end{restatable}

%% file: sections/incremental.tex
This section details our second contribution
(cf.\ Fig.~\ref{fig:overview}, \circled{3}) 
which is the definition of a valid progress measure (PM) on cofair coBüchi games which induces an incremental game solving algorithm. 

\paragraph{Valid Progress Measures}
Given a (normal/fair/cofair) game graph $\gamegraph$ a progress measure (PM) is a ranking function $\rho: V \to [0, L] \cup \{\top\}$  where $L\in \mathbb{N}$ defines the range of the PM and $[0,L] $ inherits its order from the usual order on $\mathbb{N}$, and $\top$ is the largest value where $n < \top $ for any $n \in [0, L]$, $L + 1 = \top,\top + 1 = \top$ and $\top \leq \top$. 
Two progress measures $\rho_1, \rho_2$ are compared according to element-wise comparison, i.e. $\rho_1 \preccurlyeq \rho_2$ iff for each $v \in V$, $\rho_1(v) \leq \rho_2(v)$.

A PM is called \emph{valid} for games with specification $\spec$ if it satisfies some additional rules (such as~\cref{eq:cofaircobuechi-pm} and~\eqref{eq:cofaircobuechi-prf} below), which yield a natural monotonic lifting function (such as~\cref{eq:cofaircoBuechi-lifting})  
that sends one PM $\rho$ to another $\rho' = \textsf{Lift}(\rho, v)$ that is the same as $\rho$ on all vertices and possibly \emph{increases} in exactly one vertex $v$ (which is \emph{lifted})---so always, $\rho \preccurlyeq \rho'$.
The application of this lifting function to the smallest progress measure $\rho^0 : V \to \{0\}$ (to all vertices in an arbitrary order) yields a least fixpoint (l.f.p.) PM $\rho^*$ that distinguishes the winning regions of the game:
$\Wsys(\gamegraph, \spec) = \{v \in V \mid \rho^*(v) \neq \top\}$, and consequently, $\Wenv(\gamegraph, \spec) =  \{v \in V \mid \rho^*(v) = \top\}$.

\paragraph{A Valid Cofair CoBüchi PM}
The following theorem summarizes our main contribution.

\begin{restatable}{thm}{cofaircoBuechiPM}\label{thm:cofaircoBuechiPM}
Let $\langle\gamegraph^{cf},\overline{\Inf}(\XX)\rangle$ be a cofair coBüchi game. Then the PM $\rho: V \to [0, |\XX| + |V^f|] \cup \{\top\}$ is valid for cofair coBüchi if it satisfies
\begin{equation}\label{eq:cofaircobuechi-pm}
\rho(v) \geq \begin{cases}
    \Prog(v) + 1 \quad &\text{if} \quad v\in \XX,\\
    \Prog(v) \quad &\text{if}\quad v \not\in \XX,
\end{cases}
\end{equation}
where
\begin{equation}\label{eq:cofaircobuechi-prf}
\Prog(v) = \begin{cases}
    \min \begin{cases}  
        \max_{(v,w) \in E^f} \rho(w)  \\
        \min_{(v,w) \in E}\rho(w) + 1  
    \end{cases} &\text{if} \, v \in V^f \setminus \XX, 
    \\
    \min_{(v,w) \in E}\rho(w)  &\hspace{-0.84cm}\text{if} \, v \in \Vsys  \setminus (V^f \setminus \XX), \\
    \max_{(v,w) \in E}\rho(w)  &\hspace{-0.84cm}\text{if}  \,v \in \Venv, \\
\end{cases}
\end{equation}
and it induces the lifting function
\begin{equation}\label{eq:cofaircoBuechi-lifting}
\textsf{Lift}(\rho, v) = \rho' \, \text{where} \,\,\rho'(w) = 
\begin{cases}
    \rho(w)  &\text{if}\,\quad w \neq v,\\     
    \Prog(v)+1 &\text{if}\, v = w \in \XX, \\ 
    \Prog(v) &\text{if}\, v = w \not \in \XX.
\end{cases}
\end{equation}
\end{restatable}

The proof of the above theorem is highly non-trivial and a contribution to the field of game solving. Our proof (in the appendix) uses a novel gadget-based reduction from cofair coBüchi games to (normal) coBüchi games. Gadgets, as introduced by~\citet{HausmannPSS24}, are small, DAG-like subgames %
such that by replacing a fair node with its corresponding gadget, one obtains an equivalent non-fair game that is linear in size. We present an optimized version of these gadgets, given in Fig~\ref{fig:gadgets}, which allows us to \textit{pull} the (known) coBüchi PM on the reduced coBüchi game (introduced by~\citet{Jurdzinskismall00}) \textit{back} to the cofair coBüchi game, which in turn yields the valid cofair coBüchi PM in Thm.~\ref{thm:cofaircoBuechiPM}.

\begin{figure}[htbp]
  \centering
  \begin{minipage}{0.45\linewidth}
    \centering
    \begin{tikzpicture}[auto,node distance=1.5cm,semithick]
      \node[draw, circle, minimum size=0.8cm, inner sep=0pt] (q) {$v$}; 
      \node[draw, rectangle, minimum width=0.65cm, minimum height=0.65cm, inner sep=0pt]
        (ql) [below left=2em of q] {$v_l$};
      \node[draw, circle, minimum size=0.8cm, inner sep=0pt] 
      (qr) [below right=2em of q] {$v_r$};
      \node[draw, circle, minimum size=0.6cm, inner sep=0pt] at (qr.center) {};
      \node (eleft)  [below=1.5em of ql] {$E^f(v)$};
      \node (eright) [below=1.5em of qr] {$E(v)$};
      \draw[->] (q)  -- (ql);
      \draw[->] (q)  -- (qr);
      \draw[->] (ql) -- (eleft);
      \draw[->] (qr) -- (eright);
    \end{tikzpicture}
    \par\smallskip
    
  \end{minipage}\hfill
  \begin{minipage}{0.45\linewidth}
    \centering
    \begin{tikzpicture}[auto,node distance=1.5cm,semithick]
      \node[draw, circle, minimum size=0.8cm, inner sep=0pt] (q) {$v$}; 
      \node[draw, circle, minimum size=0.6cm, inner sep=0pt] at (q.center) {};
      \node (e)  [below=1.5em of q] {$E(v)$};
      \draw[->] (q)  -- (e);
    \end{tikzpicture}
    \par\smallskip 
  \end{minipage}
  \caption{Simplified cofair coBüchi gadgets for $v \in V^f \setminus \XX$ (left) and $v \in V^f \cap \XX$ (right). Doubly encircled nodes denote the coBüchi nodes.}
  \label{fig:gadgets}
\end{figure}
\paragraph{Enhanced PM range for (the dual of) the abstract game}
Let $\mathcal{G}$ be an abstract fair Büchi game and $\overline{\mathcal{G}}$ be its dual. The range of the valid cofair coBüchi PM can be reduced to $[0, |S| + |B|]$ for $\overline{\mathcal{G}}$ due to the DAG-like structure of the 2-layers of $\Venv$ nodes in the abstract game (\cref{fig:overview} (B)). This enhancement is important as in an abstract game $|V^f|$ can be as large as $|S|^2$. %

\paragraph{Incremental Game Solving}
Recall that initializing the PM over a game to $\rho^0 : V \to \{0\}$ and iteratively applying $Lift(\cdot, \cdot)$ results in the l.f.p. $\rho^*$ which allows to extract the winning regions.
In fact, to obtain the l.f.p. $\rho^*$, we do not need to start applying the lifting function to $\rho^0$, but applying it to any PM $\rho' \preccurlyeq \rho^*$ will carry us to $\rho^*$.

It was first observed by~\citet{ChatterjeeHenzinger14} that this property can be exploited to designate extremely efficient re-solution algorithms under game graph modifications that can only shrink the \Psys winning region. This is because shrinking $\Wsys$ is correlated with increasing PM values (as can be observed from $\rho^*(v) = \top \Rightarrow v \in \Wenv$). Such graph modifications are called \emph{monotonic}. 

Take a game $\ltup{\gamegraph^{\circ}, \spec}$ with a valid PM, compute its l.f.p. PM $\rho^*$ and modify $\gamegraph^{\circ}$ via a monotonic graph modification to obtain $\gamegraph'$. We can compute the l.f.p. PM %
$\rho'^*$ of $\ltup{\gamegraph', \spec}$ by warm-starting the computation from $\rho^*$ on $G'$ and applying the lifting function until the new  l.f.p. $\rho'^*$ is reached.

We showed that new arriving data samples induce monotonic graph modifications on the dual of the abstract game (\cref{cor:monotonicity}). Therefore, the lifting function in~\cref{eq:cofaircoBuechi-lifting} defines an incremental solution algorithm 
for the dual game. 

\begin{rem}
We note that our incremental game solving technique can be extended to parity and Rabin specifications (which are generalizations of Büchi specifications) as both are positional for Player 0, and have well-defined progress measures \cite{Jurdzinskismall00,majumdar2024rabin} as well as known gadget constructions \cite{chatterjee2005complexity,HausmannPSS24}. The same incremental principles and monotonicity arguments we use for Büchi games apply, and the corresponding theorems can be reproduced using these alternative game formulations (corresponding PM and gadget constructions). However, neither will be as tractable as Büchi specifications, as neither have known polynomial solution algorithms.
\end{rem}

%% file: sections/fixpoint-algorithm.tex
The previous sections detailed 
our first two contributions (cf. Fig.~\ref{fig:overview} (A), \circled{1}\&\circled{2}) which can be combined to an algorithm that incrementally refines the winning region when new data samples arrive, due to monotonicity of learning-induced graph modifications (Fig.~\ref{fig:overview} (A) $\&$ red arrows, \cref{cor:monotonicity}). This algorithm is summarized in Alg.~\ref{alg:main_fixed}.
\begin{algorithm}[t]
\caption{Data-driven end-to-end synthesis}
\label{alg:main_fixed}
\begin{algorithmic}[1]
\REQUIRE Lipschitz constant $L_X$, noise support $[l,h]$, dataset $\mathbf{D} = \mathbf{D_N} = \{(x_i,u_i,x_i^+)\}_{i=1}^{N}$, $I \subseteq S$.
\ENSURE Current winning region $\Wsys$, a control policy \\$\pi: \Wsys \to \mathcal{U}$ if $I \cap \Wsys \neq \emptyset$
\STATE $\mathcal{G}, \rho^* \gets \texttt{initialise}(\mathbf{D}, L_X, [l,h])$;   %
\WHILE{a new sample $(x,u,x^+)$ arrives}
    \STATE $\mathbf{D}\gets\mathbf{D}\cup\{(x,u,x^+)\}$;
    \STATE update $\underline F,\overline F$ and modify $\mathcal{G}$ accordingly;
   \STATE lift $\rho^*$ on $\overline{\mathcal{G}}$ (\cref{eq:cofaircoBuechi-lifting}), until it reaches the new $\rho^*$;
   \STATE  \textbf{print } $\Wsys \gets \{s \in S \mid \rho^* = \top\}$;
    \IF{ $I \cap \Wsys \neq \emptyset$} 
    \STATE  \textbf{output } $\pi \gets\texttt{synt-controller}(\mathcal{G}, \Wsys)$
    \ENDIF
\ENDWHILE
\end{algorithmic}
\end{algorithm}
Alg.~\ref{alg:main_fixed} starts by an initialization step (line 1) where
(i) the fair Büchi game $\mathcal{G} $ is built from the initial data set $\mathbf{D_N} = \{(x_i,u_i,x_i^+)\}_{i=1}^{N}$ (cf. \circled{1}, Sec.~\ref{sec:learning}) and then 
(ii) its dual cofair coBüchi game $\overline{\mathcal{G}}$ is solved to obtain the initial l.f.p. PM $\rho^*$.
When a new (set of) sample(s) arrive (line 2), the game graph of $\mathcal{G}$ is modified according to the new approximation sets (line 4) and $\rho^*$ is lifted to the l.f.p. PM of the new $\mathcal{G}$ (line 5).

\paragraph{Initialization}
The simplest way to realize step (ii) of the initialization is to apply the lifting algorithm (\cref{eq:cofaircoBuechi-lifting}) on $\overline{\mathcal{G}}$ %
starting from $\rho^0$ until $\rho^*$ is obtained. However, this approach is known to be very time-consuming, as every node may be lifted up to $|S| + |\XX|$ times, %
 which is typically very large. 
To circumvent this problem, we use a symbolic fixpoint algorithm to solve the initial dual game $\overline{\mathcal{G}}$. \cref{eq:fixpoint-cofaircoBuechi} in~\cref{lemma:cofaircoBuechifp} is the negation of the $\mu$-calculus formula given in~\citet{banerjee2023fast} for fair Büchi games.\footnote{See~\citet{Kozen:muCalculus} for $\mu$-calculus semantics.}

\begin{restatable}{lemma}{cofaircoBuechifp}\label{lemma:cofaircoBuechifp}
Let $\mathcal{G}$ be a fair Büchi game on Büchi set $\XX \subseteq V$. Then $\psi = \Wenv(\mathcal{G}) = \Wsys(\overline{\mathcal{G}})$, where
\begin{align}\label{eq:fixpoint-cofaircoBuechi}
\psi := 
\mu Y.\, \nu X. \, &( \neg\XX \,\cup \,\Cpre(Y) ) \, \cap \, \Cpre(X) \, \cap \,\\ &( \Lpre(X) \, \cup \, \Preexists(Y)\, \cup \, \neg V^f)\notag.
\end{align}
Here, $\neg H :=V \setminus H$ for a subset $H$ of $V$. 
\end{restatable}

While solving $\overline{\mathcal{G}}$ via \eqref{eq:fixpoint-cofaircoBuechi} is computationally efficient, it does not calculate $\rho^*$ (needed to initialize Alg.~\ref{alg:main_fixed} in line 1). While it is folklore knowledge that nested $\mu$-calculus formulas yield a mapping $\rho^{\psi}: V \to \mathbb{N} \cup \{\top\}$ %
acting like a PM, there is no guarantee that $\rho^{\psi} = \rho^*$, as PMs are not unique.  %
We show (in the appendix) that indeed, $\rho^{\psi} = \rho^*$ holds.
This is not a coincidence since we have crafted the gadgets %
used to obtain the cofair coBüchi PM (\cref{thm:cofaircoBuechiPM}) to obtain this equivalence.
This equivalence allows us to combine efficient batch synthesis via a symbolic fixpoint algorithm %
with incremental game solving (cf. Fig.~\ref{fig:overview} (C), \circled{3}) and ensures the end-to-end efficiency of our method.

\paragraph{Controller Synthesis} 
As the winning region of the abstract game expands with more data becoming available, a policy that is once winning (from a state $s$) remains winning. Suppose we want to synthesize controllers for some set $I\subseteq \Wenv(\mathcal{G})$ of abstract states once they become winning.
The l.f.p. $\rho^*$ has another nice property: It yields a \Psys winning strategy that sends each $v\in \Vsys$ (when $\Vsys\cap V^f=\emptyset$) to its minimal $\rho^*$--successor, which is easily interpreted as a control policy $\pi$. %
However, solving the dual abstract game via PMs, as we do, produces an environment policy, not a controller. To obtain the latter, \texttt{synth-controller} computes the $\mu$-calculus fixpoint for the negated formula $\overline\psi$ of~\cref{eq:fixpoint-cofaircoBuechi}, warm-started with $\Wsys(\mathcal{G})$, and skips straight to the final iteration. This last iteration is sufficient to obtain $\rho^{\overline{\psi}}$ which induces the minimal-successor controller policy in a single, efficient step (as $\Vsys\cap V^f=\emptyset$ in $\mathcal{G}$).

%% file: sections/experiments.tex
\begin{figure}[t]
    \centering
    \includegraphics[width=\linewidth]{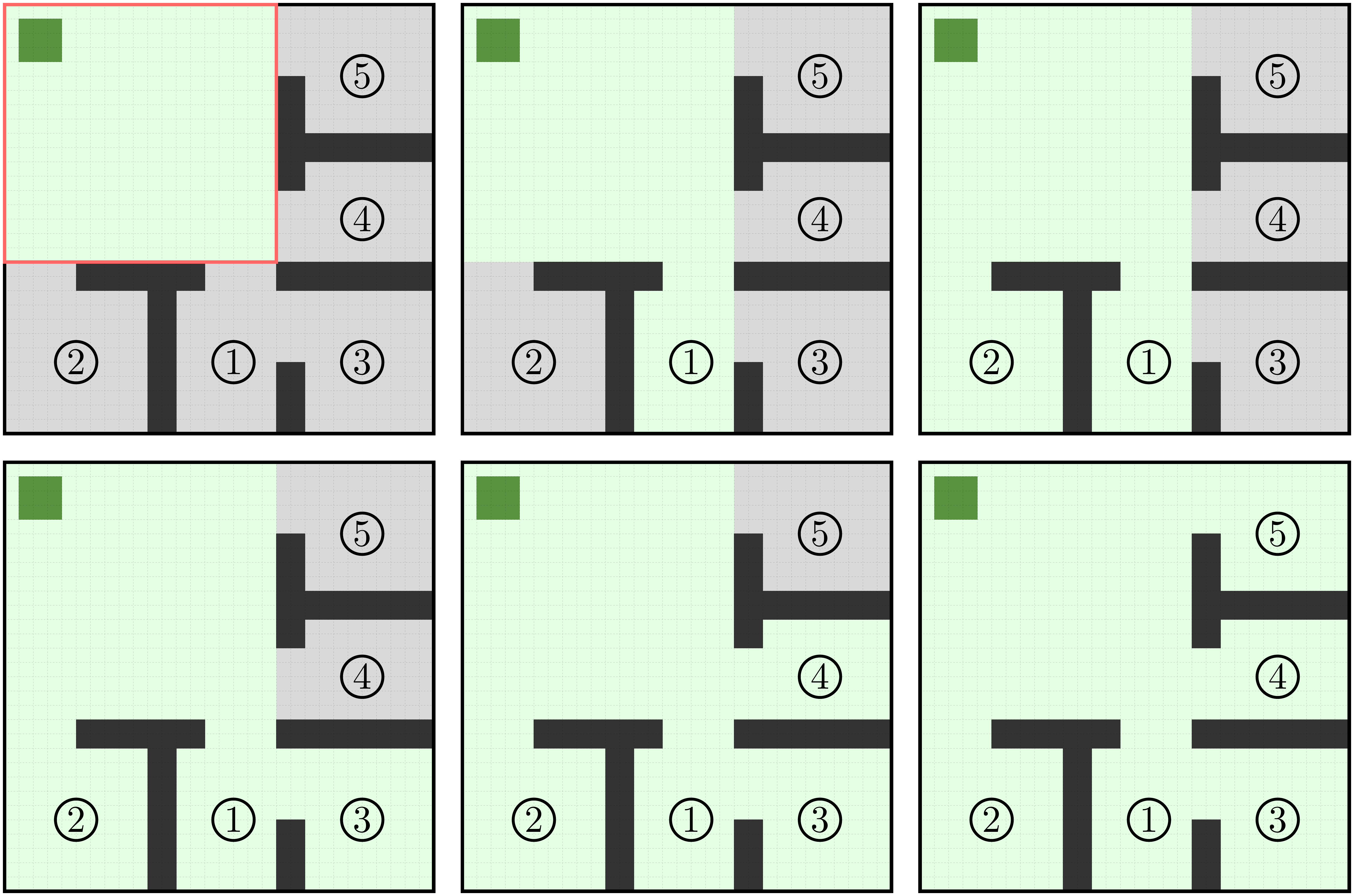}
    \caption{Gray areas indicate low-data regions and initially are not in the winning region. %
    As observations arrive from each room, they become part of the winning domain sequentially.}

    \label{fig:exp1}
\end{figure}

We evaluate our approach on a 2D robot case study adapted from~\citet{nazeri2025data}. We provide additional results in the appendix.
The specification is to visit the goal set (in dark green) infinitely often while avoiding obstacles (in black) illustrated in Fig.~\ref{fig:exp1}.
Initially, the samples are dense inside the red box (collected by the robot moving inside the red box) but sparse elsewhere, making the gray area’s under- and over-approximations overly conservative (Fig.~\ref{fig:exp1} top-left).
We first compute the winning region (light green) with a C++ symbolic fixpoint solver on a dual-core 3.3 GHz, 16 GB RAM machine. After gathering additional data in the gray zone, we re-solve the game via our incremental lifting algorithm and compare its runtime to full re-computation via the symbolic fixpoint solver (see Table~\ref{tab:runtime_comparison}). Both methods produce the same winning domain as expected, and the proposed lifting algorithm runs faster.
The enlarged winning regions upon the arrival of new data samples are indicated sequentially in Fig.~\ref{fig:exp1}.
In each update, the lifting algorithm outperforms the re-computation.
Fig.~\ref{fig:exp_res} compares the runtime required to update the winning domain as additional data is available for a $1 \times 1$ region of the state space, demonstrating consistent and substantial speedups for our approach. This experiment shows that the proposed lifting algorithm is exponentially faster in scenarios where the game graph is updated locally.

\begin{figure}[t]
    \centering
    \includegraphics[width=0.8\linewidth]{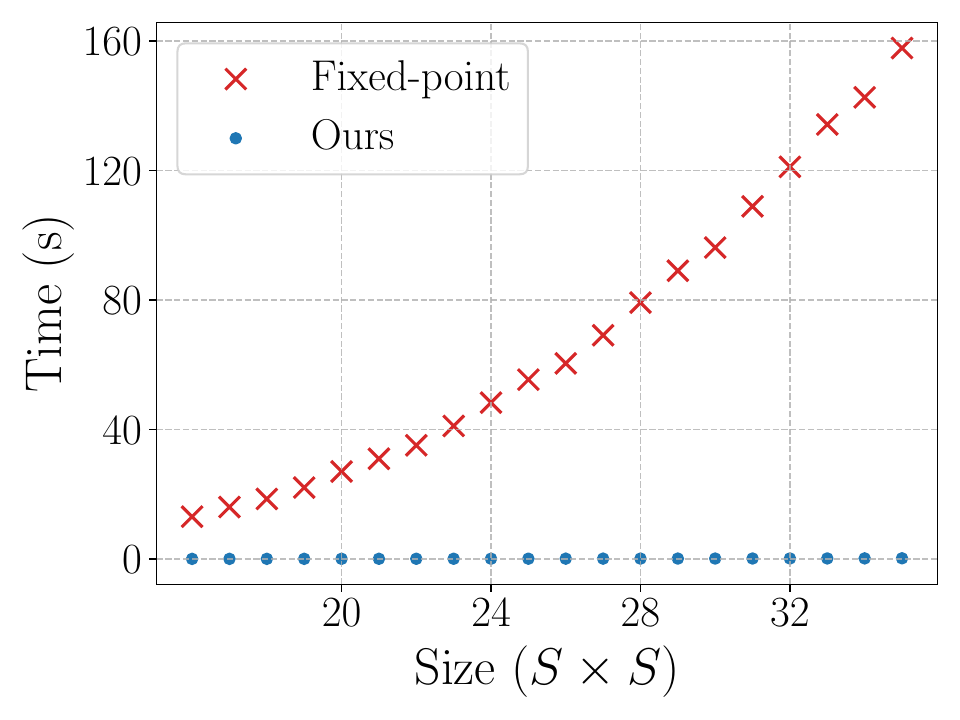}
    \includegraphics[width=0.8\linewidth]{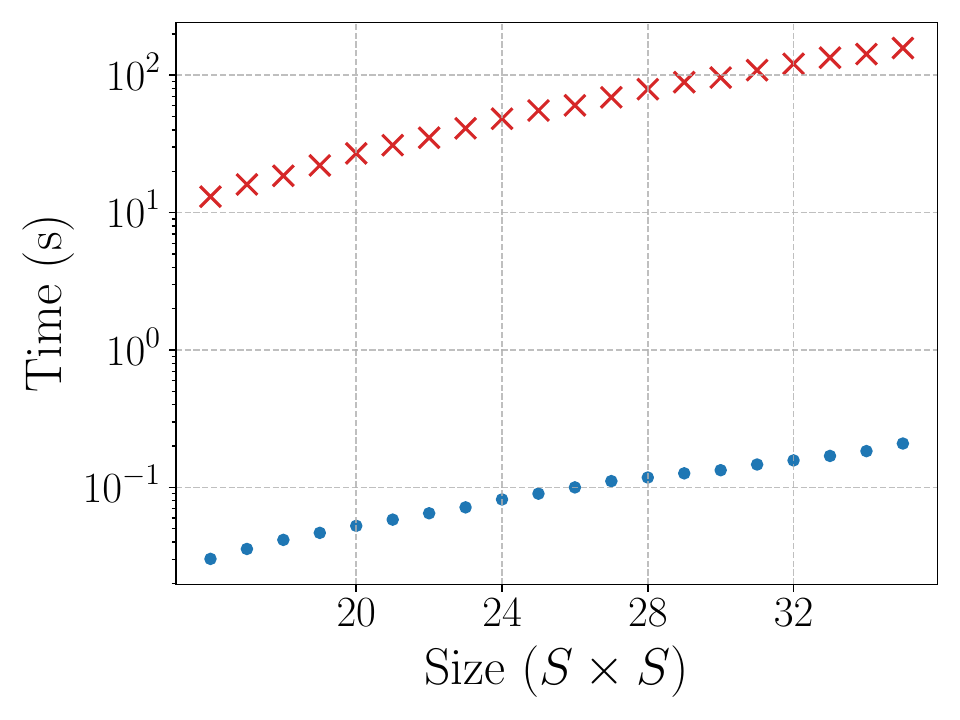}
\caption{Execution time of incremental lifting (blue dots) vs. full fixpoint recomputation (red crosses) across varying graph sizes; linear scale (top) and logarithmic scale (bottom).}
    \label{fig:exp_res}
\vspace{-0.5cm}
\end{figure}

\begin{table}[b]\label{table:comparison}
  \centering
  \resizebox{\columnwidth}{!}{%
    \begin{tabular}{ccc}
      \toprule
        Room & Incremental lifting (s) & Fixpoint recomputation (s) \\
        \midrule
        1 & 18.69 & 173.37 \\
        2 & 14.86 & 154.31 \\
        3 & 11.72 & 162.81 \\
        4 & 12.79 & 139.54 \\
        5 & 15.24 & 79.10 \\
        \bottomrule
        \end{tabular}
    }
    \caption{Runtime of incremental lifting vs.\ full fixpoint recomputation. The lifting algorithm computes the updated winning domain (which includes the new room where more data came from) faster than recomputation of the wining domain using the fixpoint solver.}
\label{tab:runtime_comparison}
\end{table}

%% file: appendix/related_work.tex
\section{Related Work}
\label{sec:related}

\paragraph{Abstractions-Based Control Policy Synthesis}
Abstraction-based control policy synthesis has become a foundational approach in formal methods and cyber-physical systems, enabling effective design and verification by simplifying complex dynamics into finite-state representations. We refer the reader to the seminal works by \citet{tabuada2009verification,girard2005approximate,reissig2016feedback} that introduce bisimulation-based symbolic models, their approximate versions, and feedback refinement relations. Extensions to stochastic dynamical systems have been extensively studied in the literature \cite{abate2008probabilistic,Soudjani2013adaptive}. As a major limitation, these approaches often require knowing an exact mathematical model of the system under study.

\paragraph{Abstraction-Based Policy Synthesis From Data}
A growing body of research has addressed \emph{data-driven} formal methods for safety verification and control synthesis, including recent abstraction-based learning methods \cite{gracia2023distributionally,banse2023data2,Chekan2023UncertainConstraints,makdesi2023data,kazemi2022datadriven,zhang2024formal,schoen2023bayesian,nazeri2025data, DBLP:journals/corr/abs-2508-15543}. Learning abstractions for dynamical systems is a powerful approach that enables synthesis and verification, without requiring exact knowledge about the system dynamics.

In deterministic settings, data-driven methods have been used to construct symbolic models from observed trajectories---enabling verification and control synthesis through automata learning, compositional abstractions, or neural approximations. For instance, \citet{banse2023data2} proposed memory-dependent Markov chain abstractions that mitigate bias via sequential sampling and adaptive partition refinement, using a Kantorovich-inspired metric for increased accuracy and scalability. Similarly, \citet{kazemi2022datadriven} developed abstraction-based control synthesis frameworks that infer growth bounds and use scenario-based convex optimization to ensure formal guarantees from sampled dynamics.

Stochastic systems introduce additional complexity due to inherent uncertainty and non-determinism. \citet{gracia2023distributionally} tackled this by learning robust Markov decision process (RMDP) abstractions from data for switched stochastic systems with unknown disturbances, combining robust-value iteration with provable high-probability satisfaction guarantees under temporal logic specifications. Similarly, \citet{nazeri2025data, DBLP:journals/corr/abs-2508-15543} extended these data-driven abstraction techniques to continuous nonlinear stochastic systems, constructing interval MDPs with PAC-style confidence bounds on transition probabilities, enabling reliable policy synthesis under uncertainty. These available approaches are either restricted to non-stochastic systems, limited to finite-horizon properties, or focused on quantitative verification and synthesis. In contrast, our approach proposed in this work addresses almost sure satisfaction of infinite-horizon properties for unknown stochastic systems. 

\paragraph{Abstraction-Free Control Policy Synthesis}
Another line of research focuses on synthesizing control policies directly via certificates, bypassing the need for abstraction.
For data-driven control barrier certificates (CBCs), many existing approaches are limited to linear or control-affine dynamics~\cite{Jagtap2020CBCGP,Cohen2022,Lopez2022uCBF}.
Furthermore, many approaches rely on known Lipschitz constants to provide formal guarantees and/or address partially unknown dynamics.
For instance, Gaussian processes are employed by \citet{Wang2018CBF} and \citet{Jagtap2020CBCGP} to learn the partially unknown dynamics of nonlinear systems while assuming the affine control-dependent part of the dynamics is known.
Systems with unknown additive disturbance are addressed by \citet{Chekan2023UncertainConstraints}.
Similarly, \citet{mathiesen2024data} assume the deterministic part of the dynamics to be accurately known, whereas the CBC-based safety control method proposed by \citet{mazouz2024data} focuses on known noise distributions.
Approaches to fully unknown dynamics are scarce. \citet{Salamati2021DDCBC} study the computation of CBCs for fully unknown discrete-time systems relying on known Lipschitz constants.
Uncertain continuous-time systems are studied by \citet{wang2023stochastic}, employing Bayesian inference and local Lipschitzness. These approaches are developed for basic safety and reachability properties and cannot handle the class of specifications considered in our work.

\paragraph{Strategy Synthesis Over Game Abstractions}
Two-player games on graphs \cite{gamesongraphs}---originally used to abstract the computation flow of reactive programs in computer science---provide a powerful model for dynamic strategic interactions. 
These models enable the use of \textbf{reactive synthesis} \cite{Bernd-lecturenotes} to design control strategies that continuously react to actions of the environment to satisfy temporal logic specifications. Such specifications are widely used to formalize both safety and strategic liveness objectives for cyber-physical systems, e.g., for mobile robot navigation \cite{KressGazitFainekosPappas2009,kress2018synthesisForRobotsReview} or in autonomous driving~\cite{AlthoffBeltaReviewFMCEforAutonomousDriving}. 

Abstraction-based control policy synthesis (cf. previous paragraphs of this section) constructs abstractions as games and utilizes reactive synthesis for strategy construction. %
\citet{majumdar2024symbolic} first observed that abstracting into \emph{fair} games significantly reduces conservatism, leading to successful synthesis in broader contexts. In particular, it has been shown that fair games have many interesting computational advantages \cite{banerjee2023fast, SaglamS23, HausmannPSS24, SchmuckTSN24, AnandNRSS25}, which makes their use in abstraction-based policy synthesis very promising. %

\paragraph{Incremental Game Solving}
The first incremental/decremental game solving algorithm was proposed by \citet{ChatterjeeHenzinger14} for Büchi/coBüchi games, which relies on progress measures (PMs). First presented in~\citet{KK91}, progress measures became more known for Jurdzinski's presentation~\cite{Jurdzinskismall00} for parity games, under the name ``small progress measures.''

Very recently, the incremental algorithm by \citet{ChatterjeeHenzinger14} was extended to fair Büchi games by \citet{SaglamST24}. Although conceptually related, their approach differs significantly from the one presented in this paper. While \citet{SaglamST24} construct PMs directly and focus on heuristics for non-monotone graph modifications, we utilize a simplified version of the gadgets from \citet{HausmannPSS24} to construct cofair coBüchi PMs; that are aligned with a symbolic fixed-point algorithm adapted from the fair Büchi game solver proposed by \citet{banerjee2023fast}. 

To the best of the authors' knowledge, this is the first work to connect monotonicity in data-driven abstraction refinement with the monotonicity required for incremental game solving.

%% file: appendix/app-notation.tex
\section{Notation}\label{app:notation}
This section provides the notation necessary to follow the proofs and the additional material presented in this appendix.

\paragraph{Basics}
We use $\mathbb{N}$ and $\mathbb{R}$ to denote the set of non-negative integers and real numbers, respectively. Let $\Sigma$ be a finite alphabet. Then  $\Sigma^*$ and $\Sigma^\omega$ denote the sets of finite and infinite words over $\Sigma$. 
For an infinite word $\xi \in \Sigma^\omega$, we denote by $\Inf(\xi)$ the set of letters occurring infinitely often in $\xi$.

We write $\|v\|$ for the $L^\infty$ norm of vector $v$, and $|v|$ for its element-wise absolute value. 
We use $\preccurlyeq$ and $\succcurlyeq$ to denote element-wise inequalities for vectors in $\mathbb{R}^n$. Specifically, for $l, h \in \mathbb{R}^n$, we write $l \preccurlyeq h$ if and only if $l(i) \leq h(i)$ for all $i \in \{1, 2, \ldots, n\}$, where $l(i) \in \mathbb{R}$ denotes the $i$-th component of the vector $l$.
We use $[1, n]$ to denote the set of natural numbers from $1$ to $n \in \mathbb{N}$, inclusive.
For vectors $l, h \in \mathbb{R}^n$ with $l \preccurlyeq h$, the notation $[l, h]$ denotes the axis-aligned hyperrectangle (box) defined by the Cartesian product of intervals, $[l, h] := \prod_{i=1}^n [l(i), h(i)]$. 

\input{sections/game-prelims}

%% file: sections/game-prelims.tex
\paragraph{Game graphs}
A \emph{game graph} is a tuple $\gamegraph= \tup{V,\Vsys,\Venv,E}$, where $(V,E)$ is a finite directed graph with  \emph{edges} $ E $ and \emph{vertices} (\emph{nodes}) $ V $ partitioned into \Psys and \Penv vertices, $\Vsys$ and $\Venv$, respectively. We refer to such game graphs as \emph{non-fair} game graphs. For $v \in V$, $E(v)$ denotes its successor set $\{w \mid (v, w) \in E\}$. 
Without loss of generality, we assume all vertices in $V$ have at least one outgoing edge. 
A \emph{play} originating at a vertex $v_0$ is an infinite sequence of vertices $\xi=v_0v_1\ldots \in V^\omega$ with $v_{i+1}\in E(v_i)$ for all $i\in\mathbb{N}$. 
We denote the restriction of a game graph $\gamegraph$ to a subset $V' \subseteq V$ as $\gamegraph|_{V'} = \tup{V\cap V', \Vsys\cap V', \Venv \cap V', E \cap V' \times V'}$. In all figures, circle nodes denote $\Vsys$ vertices and square nodes denote $\Venv$ vertices.

\paragraph{Fair and cofair game graphs}
A \emph{fair game graph} is a tuple $\gamegraph^f= (\gamegraph, E^f)$
where $\gamegraph = \tup{V,\Vsys,\Venv,E}$ is a game graph and $E^f \subseteq E \cap \Venv \times V$ is a subset of edges originating from $\Venv$ vertices. 
We call a game graph \emph{cofair} if $E^f$ originates from $\Vsys$ vertices instead, and denote it by $\gamegraph^{cf}$. 
We call $E^f$ \emph{fair edges} and nodes with fair edges \emph{fair vertices}, $V^f = \{v \in V \mid E^f(v) \neq \emptyset \}$. A play $\xi$ in $\gamegraph$ is called \emph{fair} in $\gamegraph^f$ if in $\xi$ \emph{whenever a fair vertex is visited infinitely often, all its fair outgoing edges are taken infinitely often as well}. Formally, $\xi$ is fair iff for every $v \in V^f$, $v \in \Inf(\xi) \implies \forall w \in E^f(v), w \in \Inf(\xi)$.

\paragraph{Games}
A game $\mathcal{G} = \ltup{\gamegraph^\circ, \spec}$ consists of a (non-fair, fair, or cofair) game graph and a specification. A game is called a non-fair/fair/cofair depending on the game graph.

In a non-fair game $\ltup{\gamegraph, \spec}$, a play $\xi$ is winning for $\Vsys$ iff $\xi \in \spec$. 
In a fair game $\ltup{\gamegraph^f, \spec}$ (fair nodes belong to \Penv), a play $\xi$ is winning for $\Vsys$ iff $\xi \in \spec$ \emph{or} $\xi$ \emph{is not} fair.
In a cofair game $\ltup{\gamegraph^{cf}, \spec}$ (fair nodes belong to \Psys), a play $\xi$ is winning for $\Vsys$ iff $\xi \in \spec$ \emph{and} $\xi$ \emph{is} fair. Otherwise, $\xi$ is winning for $\Venv$. 

\paragraph{Büchi and coBüchi games}
In this work, we restrict our attention to specifications 
$\spec = \Inf(\XX)$ (called a Büchi specification) and $\spec = \overline{\Inf}(\XX)$ (coBüchi specification). The set $\XX \subseteq V$ is called Büchi vertices in the first specification and coBüchi vertices in the second. $\Inf(\XX)$ defines all plays in $V^\omega$ that visit a node in $\XX$ infinitely often, and $\overline{\Inf}(\XX)$ defines all plays in $V^\omega$ that do not visit any node in $\XX$ infinitely often.
Games with Büchi/coBüchi specifications are called (non-fair/fair/cofair) Büchi/coBüchi games, depending on the game graph.

\paragraph{Strategies}
A \emph{strategy} for Player $j \in \{0,1\}$ over the game graph $\gamegraph^\circ$ %
is a function $\sigma^j : V^* \cdot V_j \to V$ with the constraint that for all $h\cdot v \in V^* \cdot V_j$ it holds that $\sigma^j(h \cdot v) \in E(v)$. A play $\xi=v_0 v_1\ldots \in V^\omega$ is called a $\sigma^j$-play if for all $i\in \mathbb{N}$ holds that $v_i\in V_j$ implies $v_{i+1}=\sigma^j(v_0 \ldots v_i)$. A strategy $\sigma^j$ is winning for Player $j$ from a subset $V' \subseteq V$ in the game $\ltup{\gamegraph^\circ,\spec}$, if all $\sigma^j$-plays $\xi$ in $G^\circ$ that start at a vertex in $V'$ are winning for the player w.r.t.\ $\spec$. 
The winning region of Player $j$ is the maximal set $V'\subseteq V$ that is winning for Player $j$, which is denoted by $\Win_j(\gamegraph, \spec)$, or simply $\Win_j$ whenever the game is clear from the context.
For the games we are interested in, it always hold that $V = \Wsys \cupdot \Wenv$.
A strategy $\sigma$ is called \emph{positional} iff for all $h, h' \in V^*$, $\sigma(h \cdot v) = \sigma(h' \cdot v)$. In this case, we simply denote the strategy at vertex $v$ by $\sigma(v)$.

Both Büchi and coBüchi games accept positional strategies; that is, both players have positional winning strategies in their winning regions. In fair/cofair Büchi and coBüchi games, the player without the fair edges has positional strategies from its winning region, whereas the owner of the fair vertices might need exponential memory strategies~\cite{SaglamS23}.

\paragraph{Duality of Büchi and coBüchi Games} Büchi and coBüchi games are dual in the following sense: 
In a Büchi game $\ltup{\gamegraph, \Inf(\XX)}$, Player 0 is playing Büchi game on the set $\XX$ (i.e. trying to ensure $\Inf(\XX)$) and Player 1 is playing a coBüchi game on the set $\XX$ (i.e. trying to ensure $\overline{\Inf}(\XX)$).

Consequently, if we have a Büchi game $\ltup{\gamegraph, \Inf(\XX)}$ with $\gamegraph= \tup{V,\Vsys,\Venv,E}$, we can swap the ownership of the nodes and turn the Büchi nodes into coBüchi nodes, and get an \emph{equivalent} coBüchi game $\ltup{\overline{\gamegraph}, \overline{\Inf}(\XX)}$ where $\overline{\gamegraph} = \tup{V,\overline{\Vsys} = \Venv,\overline{\Venv} = \Vsys,E}$. These two games, which are called \textit{duals} of each other, are equivalent in the sense that $\Wsys(\gamegraph, \Inf(\XX))  =\Wenv(\overline{\gamegraph}, \overline{\Inf}(\XX))$, and consequently, $\Wenv(\gamegraph, \Inf(\XX))  =\Wsys(\overline{\gamegraph}, \overline{\Inf}(\XX))$. 

Similarly, fair Büchi and cofair coBüchi games are dual. Let $\mathcal{G} = \ltup{\gamegraph^f, \Inf(\XX)}$ be a fair Büchi game with $\gamegraph^f=(\gamegraph, E^f)$. Then the cofair coBüchi game $\overline{\mathcal{G}} = \ltup{\gamegraph^{cf}, \overline{\Inf}(\XX)}$ where $\gamegraph^{cf}=(\overline{\gamegraph}, E^f)$ (note that this game graph is cofair since in $\overline{\gamegraph}$, the fair edges $E^f$ originate from $\Vsys$ nodes) is its dual, i.e. $\Wsys\ltup{\mathcal{G}} = \Wenv\ltup{\overline{\mathcal{G}}}$. Observe that by solving a game, we automatically obtain the solution to its dual.

%% file: appendix/app-sec3.tex
\section{Supplementary Material for Sec. 3}
\subsection{Proof of Equation~\ref{eq:bounds}}
\begin{proof}[\unskip\nopunct]
Fix an arbitrary state $x_* \in \mathbb{R}^n$ and control input $u_* \in \mathcal{U}$. Consider the subset of data samples $\mathbf{D_N}(u_*) \subseteq \mathbf{D_N}$ defined as
$$
\mathbf{D_N}(u_*) = \{(x_i, u_i, x_i^+) \in \mathbf{D_N} \mid u_i = u_*\}.
$$
Given any sample $(x_i, u_*, x_i^+) \in \mathbf{D_N}(u_*)$, we have
$$
x_i^+ = f(x_i, u_*) + w_i.
$$
Since the noise term $w_i$ is unknown but bounded within the compact set $\Omega = [l,h]$, we derive the following bounds for $f(x_i, u_*)$:
\begin{equation*}
x_i^+ - h \preccurlyeq f(x_i, u_*) \preccurlyeq x_i^+ - l.
\end{equation*}
Using the Lipschitz continuity, we have
\begin{equation*}
|f(x_i, u_*) - f(x_*, u_*)| \preccurlyeq L_X\cdot\|x_i - x_*\|.
\end{equation*}
Combining the two inequalities above, we obtain
\begin{equation*}
x_i^+ - L_X\cdot\|x_i - x_*\| - h \preccurlyeq f(x_*, u_*) \preccurlyeq x_i^+ + L_X\cdot\|x_i - x_*\| - l.
\end{equation*}
Since this holds for all samples $(x_i, u_*, x_i^+) \in \mathbf{D_N}(u_*)$, we combine these inequalities into a single expression by introducing functions:
\begin{align}
\check{f}(x_*, u_*|\mathbf{D_N}) := &\max_{\substack{(x_i, u_*, x_i^+) \\ \in \mathbf{D_N}(u_*)}}
\left(x_i^+ - L_X\cdot\|x_i - x_*\|\right) - h,\notag\\
\hat{f}(x_*, u_*|\mathbf{D_N}) := &\min_{\substack{(x_i, u_*, x_i^+) \\ \in \mathbf{D_N}(u_*)}}
\left(x_i^+ + L_X\cdot\|x_i - x_*\|\right) - l,\notag
\end{align}
where we have
\begin{equation*}
    f(x_*, u_*) \in \Big[\check{f}(x_*, u_*|\mathbf{D_N}), \hat{f}(x_*, u_*|\mathbf{D_N})\Big].
\end{equation*}
\end{proof}

\subsection{Proof of~\cref{thm:ouapproximation}}
\ouapproximation*

\begin{proof}[\unskip\nopunct]
\textbf{Data-driven over-approximation} We show that any abstract state $s'$ that is reachable from a concrete state $x\in\mathcal{T}^{-1}(s)$ with a non-zero probability, i.e., $T(s' \mid x, u) > 0$, is contained within the data-driven over-approximation $\overline{F}(s, u|\mathbf{D_N})$.

Consider abstract states $s, s' \in S$, a concrete state $x \in \mathcal{T}^{-1}(s)$, and a control input $u \in \mathcal{U}$ satisfying
$$ T(s' \mid x, u) > 0. $$
This implies the existence of a noise realization $w \in \Omega$ such that
$$ f(x, u) + w \in \mathcal{T}^{-1}(s'), $$
which is equivalent to
$$ \{w \in \Omega \mid f(x, u) + w \in \mathcal{T}^{-1}(s')\} \neq \emptyset. $$
Given the support of the noise $\Omega = [l, h]$, we write this as
$$ \Big[f(x, u) + l, f(x, u) + h\Big] \cap \mathcal{T}^{-1}(s') \neq \emptyset. $$
Using data-driven bounds in \cref{eq:bounds}, we obtain
$$ \left[\check{f}(x, u|\mathbf{D_N}) + l, \hat{f}(x, u|\mathbf{D_N}) + h\right] \cap \mathcal{T}^{-1}(s') \neq \emptyset. $$
We can expand the first term and obtain
\begin{multline*}
\left[\min_{x \in \mathcal{T}^{-1}(s)} \check{f}(x, u|\mathbf{D_N}) + l, \max_{x \in \mathcal{T}^{-1}(s)} \hat{f}(x, u|\mathbf{D_N}) + h\right] \\ \cap \mathcal{T}^{-1}(s') \neq \emptyset,
\end{multline*}
which, by definition, implies that $s' \in \overline{F}(s, u|\mathbf{D_N})$.

\bigskip
\textbf{Data-driven under-approximation} We show that for every abstract state $s' \in \underline{F}(s, u|\mathbf{D_N})$, any concrete state $x \in \mathcal{T}^{-1}(s)$ can reach $s'$ with a non-zero probability, i.e., $T(s' \mid x, u) > 0$.

Consider an abstract state $s' \in \underline{F}(s, u|\mathbf{D_N})$. By the definitions of $\underline{F}(s, u|\mathbf{D_N})$ and $\underline{f}(s, u|\mathbf{D_N})$, we have
\begin{multline*}
\mathcal{T}^{-1}(s') \subseteq \bigg[\max_{x \in \mathcal{T}^{-1}(s)} \hat{f}(x, u|\mathbf{D_N}) + l, \\ \min_{x \in \mathcal{T}^{-1}(s)} \check{f}(x, u|\mathbf{D_N}) + h\bigg].
\end{multline*}
Using \cref{eq:bounds}, we get
\begin{align*}
&\left[\max_{x \in \mathcal{T}^{-1}(s)} \hat{f}(x, u|\mathbf{D_N}) + l, \min_{x \in \mathcal{T}^{-1}(s)} \check{f}(x, u|\mathbf{D_N}) + h\right] \\
&\quad\subseteq \left[\max_{x \in \mathcal{T}^{-1}(s)} f(x, u) + l, \min_{x \in \mathcal{T}^{-1}(s)} f(x, u) + h\right] \\
&\quad\subseteq \bigcap_{x \in \mathcal{T}^{-1}(s)} [f(x, u) + l, f(x, u) + h].
\end{align*}
Therefore, it follows that
$$\mathcal{T}^{-1}(s') \subseteq \bigcap_{x \in \mathcal{T}^{-1}(s)} \Big[f(x, u) + l, f(x, u) + h\Big] $$
which, by definition, implies that
$$ T(s' \mid x, u) > 0, \quad \forall x \in \mathcal{T}^{-1}(s). $$
\end{proof}

\subsection{Abstract Fair Game Graph}
~\cref{fig:gameabstraction} shows the abstract game graph constructed formally in Sec. 4. Each abstract state is a $\Vsys$ node (i.e., $s$). With each control input $u \in \mathcal{U}$, it goes to a $\Venv$ node (i.e., $s^u$). Each of these $\Venv$ nodes has $(m^{s,u}+1)$ successors, all of which are fair nodes, i.e., only have fair outgoing edges. The topmost one (i.e., $s^u_0$) has fair outgoing edges to the under-approximation $\underline{F}(s, u)$. The rest of them also have an additional fair edge to a node in $\overline{F}(s, u) \setminus \overline{F}(s, u)$. This construction mimics the following: When the controller choses an action $u$ from an abstract state $s$, it can end up in \emph{any} abstract state $\overline{s}$ in the over-approximation $\overline{F}(s, u)$ (this is guaranteed by the existence of a $\Venv$ node $s^a_i$ that has $\overline{s}$ as a successor, so $\Venv$ has a strategy to reach $\overline{s}$ from $s$) and if it takes the input $u$ from $s$ infinitely often, it is guaranteed to end up in \emph{every} abstract state in the under-approximation $\underline{F}(s, u)$ (this is guaranteed as all successors $s^u_i$ of $s^u$ have \emph{every} state $\underline{s} \in \underline{F}(s, u)$ as a \emph{fair} successor, so every play that sees $s$ infinitely often, also sees $\underline{s}$ infinitely often).

\begin{figure}[htbp]
  \centering
    \begin{tikzpicture}[auto,node distance=1.5cm,semithick]%
    \node[draw, circle, minimum size=0.8cm, inner sep=0pt] (q) {$s$}; 
        \node[draw, rectangle, minimum width=0.65cm, minimum height=0.65cm, inner sep=0pt]        (q1) [right=2em of q] {$s^u$};
        \node[draw, rectangle, minimum width=0.65cm, minimum height=0.65cm, inner sep=0pt]
        (q3) [right=2em of q1] {$s^u_1$};
        \node[draw, rectangle, minimum width=0.65cm, minimum height=0.65cm, inner sep=0pt]
        (q2) [above=1em of q3] {$s^u_0$};
        \node (blank2)  [below=1em of q3] {$\vdots$};
        \node[draw, rectangle, minimum width=0.65cm, minimum height=0.65cm, inner sep=0pt]
        (q4) [below=1em of blank2] {$s^u_{m}$};
         \node[draw, circle, minimum size=0.8cm, inner sep=0pt]
        (q5) [right=3em of q3] {$\overline{s}^{s,u}_1$};
        \node[draw, circle, minimum size=0.8cm, inner sep=0pt]
        (q6) [right=3em of q4] {$\overline{s}^{s,u}_m$};
        \node (blank)  [below=1.5em of q1] {$\vdots$};
        \node[draw,  shape=ellipse,
       minimum width=0.8cm,
        minimum height=0.8cm, inner sep=0pt,  line width=1pt] (F) [right=2em of q2] {$\underline{F}(s, u)$}; 
        \draw[->] (q)  -- (q1);
        \draw[->] (q1)  -- (q2);
        \draw[->] (q1)  -- (q3);
        \draw[->] (q1)  -- (q4);
        \draw[->] (q1)  -- (blank2);
        \draw[->, dashed] (q2)  -- (F);
        \draw[->, dashed] (q3)  -- (F);
        \draw[->, dashed] (q4)  -- (F);
        \draw[->, dashed] (q3)  -- (q5);
        \draw[->, dashed] (q4)  -- (q6);
        \draw[->] (q)  -- (blank);
    \end{tikzpicture}
    \par\smallskip
  \caption{The figure represents the outgoing edges of a state $s$, and the subgame created thus, in the abstract fair game graph. The fair edges are denoted by dashed lines. Here, $m = m^{s,u}$ and the fair outgoing edge to the elliptic node represents fair outgoing edges to each state $\underline{s} \in \underline{F}(s,u) \subseteq V_0$.}
    \label{fig:gameabstraction}
\end{figure}
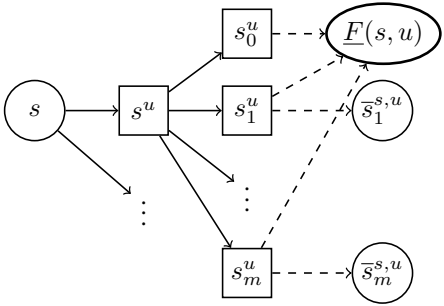

\paragraph{The running example} The abstract game graph depicted in~\cref{fig:overview} (B) is an example of this construction. Here $\mathcal{U} = \{\texttt{fw}\}$ and $S$ is the 9-cells in the grid figure. However, we treat the two wall states $\frownie$ as one for simplicity, and only show the abstract states $v_0, v_1, \smiley, \frownie$ in the abstract game graph, as the other abstract states in the grid (the white ones) are nor reachable from $v_0$ under the current abstraction sets. 
The initial game (both black and purple edges) is constructed using the approximation sets $\overline{F}^{N}(s, \texttt{fw})$, $\underline{F}^{N}(s, \texttt{fw})$ where 
for $s \in \{v_0, v_1\}$ the approximation set values are given on the figure and for $v \in \{\smiley, \frownie\}$ we assume $\overline{F}^{N}(s, \texttt{fw}) = \underline{F}^{N}(s, \texttt{fw} ) = s$ and simplify the resulting transitions in the abstract fair game with self loops. 

As $|\mathcal{U}| = 1$, $s \in \{v_0, v_1\}$ has one successor, $s^\texttt{fw} \in \Venv$. 
Since $|\overline{F}^N(s, \texttt{fw}) \setminus \underline{F}^N(s, \texttt{fw})| + 1$ is $2$ and $3$ for $s= v_0$ and $s = v_1$; $(v_0)^\texttt{fw}$ and $(v_1)^\texttt{fw}$ have $2$ and $3$ successors, respectively.  $(v_0)^\texttt{fw}_0$ and  $(v_1)^\texttt{fw}_0$---the topmost successors of  $(v_0)^\texttt{fw}$ and  $(v_1)^\texttt{fw}$ respectively---have one successor each, that is the unique abstract state in their under-approximations $\underline{F}^N(s, \texttt{fw})$ for $s = v_0$ and $s = v_1$. The other successors of $s^\texttt{fw}$ go to all of the under-approximation states, in addition to one more state from the over-approximation each. 

After one learning step, we assume the over-approximation $\overline{F}^{N+1}(v_1, \texttt{fw})$ has shrunk to $\{v_1, \smiley\}$; and everything else stayed the same. Then, the abstract game graph obtained by the under- and over-approximations  $\overline{F}^{N}(s, \texttt{fw}) = \underline{F}^{N}(s, \texttt{fw} )$ is the black subgraph; and it is easily obtained by removing the purple edges (and the now redundant node $(v_1)^\texttt{fw}_1$). In the next section, we detail how to obtain the new subgraph after a learning step by locally modifying the previous graph. 

%% file: appendix/app-sec4.tex
\section{Supplementary Material for Sec. 4}

\subsection{Proof of~\cref{thm:mono}}
\mono*
\begin{proof}[\unskip\nopunct]
We start from the definitions of data-driven bounds,
\begin{align}
\check{f}(x_*, u_* \mid \mathbf{D}_N) &= \max_{\substack{(x_i, u_*, x_i^+) \\ \in \mathbf{D}_N(u_*)}}
\left(x_i^+ - L_X \cdot \|x_i - x_*\|\right) - h, \notag\\
\hat{f}(x_*, u_* \mid \mathbf{D}_N) &= \min_{\substack{(x_i, u_*, x_i^+) \\ \in \mathbf{D}_N(u_*)}}
\left(x_i^+ + L_X \cdot \|x_i - x_*\|\right) - l. \notag
\end{align}
Consider an additional observation
$$
\mathbf{D}_{N+1}(u_*) = \mathbf{D}_N(u_*) \cup \{(x_{N+1}, u_*, x_{N+1}^+)\}.
$$
Then, the updated lower bound becomes
\begin{multline*}
\check{f}(x_*, u_* \mid \mathbf{D}_{N+1}) = \\
\max \Big( \check{f}(x_*, u_* \mid \mathbf{D}_N),\; x_{N+1}^+ - L_X \cdot \|x_{N+1} - x_*\| - h \Big),
\end{multline*}
which results in
$$
\check{f}(x_*, u_* \mid \mathbf{D}_{N+1}) \succcurlyeq \check{f}(x_*, u_* \mid \mathbf{D}_N).
$$
Similarly, the upper bound becomes
\begin{multline*}
\hat{f}(x_*, u_* \mid \mathbf{D}_{N+1}) = \\
\min \Big( \hat{f}(x_*, u_* \mid \mathbf{D}_N),\; x_{N+1}^+ + L_X \cdot \|x_{N+1} - x_*\| - l \Big),
\end{multline*}
and therefore
$$
\hat{f}(x_*, u_* \mid \mathbf{D}_{N+1}) \preccurlyeq \hat{f}(x_*, u_* \mid \mathbf{D}_N).
$$

This establishes that the lower bound $\check{f}$ is monotonically non-decreasing, while the upper bound $\hat{f}$ is monotonically non-increasing.
\end{proof}
\subsection{Proof of~\cref{thm:UAOA}}
\UAOA*
\begin{proof}
Theorem~\ref{thm:mono} establishes that
$$\check{f}(x_*, u_*|\mathbf{D_N}) \preccurlyeq \check{f}(x_*, u_*|\mathbf{D_{N+1}}) \text{ and}$$
$$\hat{f}(x_*, u_*|\mathbf{D_N}) \succcurlyeq \hat{f}(x_*, u_*|\mathbf{D_{N+1}}).$$
As a result, the learned interval for $f(x_*, u_*)$ becomes tighter with more data, and therefore,
$$
\overline{f}(s, u|\mathbf{D_{N+1}}) \subseteq \overline{f}(s, u|\mathbf{D_{N}}) \text{ and} $$
$$\underline{f}(s, u|\mathbf{D_{N}}) \subseteq \underline{f}(s, u|\mathbf{D_{N+1}}).
$$
Consequently, the corresponding reachable sets also improve monotonically, 
$$
\overline{F}(s, u|\mathbf{D_{N+1}}) \subseteq \overline{F}(s, u|\mathbf{D_{N}}) \text{ and}$$
$$\underline{F}(s, u|\mathbf{D_{N}}) \subseteq \underline{F}(s, u|\mathbf{D_{N+1}}).
$$
\end{proof}
\subsection{Monotonic Graph Modifications and Proof of~\cref{thm:dualmonotonicitya}}

\paragraph{Monotonic Game Updates.}
Let $\mathcal{G}_N = \langle\gamegraph^f_N, \Inf(\XX)\rangle$ and
$\mathcal{G}_{N+1} = \ltup{\gamegraph^f_{N+1}, \Inf(\XX)}$
be the abstract fair Büchi games induced by the under- and over-approximations  $\underline{F}^{N}(s, u)$, $\overline{F}^N(s, u)$ and $\underline{F}^{N+1}(s, u)$, $\overline{F}^{N+1}(s, u)$, respectively.

We show that the resulting modifications from $\gamegraph^f_N$ to $\gamegraph^f_{N+1}$ can only \textbf{expand} the $\Vsys$ winning region in the corresponding fair Büchi game.

As mentioned in~\cref{thm:UAOA}, at each learning step, for a state and input pair $(s, u)$, two atomic changes may occur (a) the under-approximation expands and (b) the over-approximation shrinks. It is easy to see that the under-approximation is always contained in the over-approximation, i.e., $\underline{F}^N(s, u) \subseteq \overline{F}^N(s,u)$ for all $N$.

We will examine how the atomic changes (a) and (b) alters the game graph. 

\begin{enumerate}[label=(\alph*)]
\item The under-approximation expands, resulting in the new under-approximation $\underline{F}^{N+1}(s, u) = \{s'\} \cup  \overline{F}^N(s, u) $ where $s' \in \overline{F}^{N}(s, u) \setminus \underline{F}^{N}(s, u)$.
Then the new game graph $\gamegraph_{N+1}$ has additional (fair) outgoing edges from each $s^u_i$ to $s'$, for each $i \in [1, m^{s, u}]$. \label{item:UAexpands}
\item The over-approximation shrinks, resulting in the new over-approximation 
$\overline{F}^{N+1}(s, u) = \overline{F}^{N}(s, u) \setminus \{s'\}$ for a $s' \in \overline{F}^{N}(s, u) \setminus \underline{F}^{N}(s, u) $. 
Then the new game graph $\gamegraph_{N+1}$ loses the node $s^u_i$, which had fair outgoing edges to $\underline{F}^N(s, u) \cup \{s'\}$.\label{item:OAshrinks}
\end{enumerate}
\cref{thm:dualmonotonicitya} states that, under both (a) and (b), the winning region of $\Vsys$ in $\mathcal{G}^f_{N+1}$ contains the winning region of \Psys in $\mathcal{G}^f_N$, with the exception of the lost states in modification~\ref{item:OAshrinks}.

\dualmonotonicitya*

\emph{Proof sketch. }We show that neither modifications~\ref{item:UAexpands} nor~\ref{item:OAshrinks} can shrink the winning region of \Psys. As every modification consists of (a) and (b), this proves the theorem. 

It is easy to see that~\ref{item:OAshrinks} cannot cause $\Wsys$ to shrink, because \Penv is losing the (non-fair) edge $s^u \to s^u_i$. This change is only restricting the actions of \Penv, and therefore $\Wenv$ cannot expand. Consequently, $\Wsys$ cannot shrink. 

The reason why~\ref{item:UAexpands} cannot cause $\Wsys$ to shrink is more subtle. As \Penv is gaining new edges $s^u_i \to s'$, it may appear that the strategy space of \Penv is expanding. However, firstly, although \Penv has new strategies on $\gamegraph^f_{N+1}$ that do not exist on $\gamegraph^f_{N}$, some of its strategies that were fair on $\gamegraph^f_{N}$ are not fair---and therefore losing---on $\gamegraph^f_{N+1}$, because the newly added edges are fair, and on $\gamegraph^f_{N+1}$, whenever $s \to s^u_i$ is seen infinitely often for some $s^u_i$, $s^u_i \to s'$ should also be taken infinitely often (which was not a necessity on $\gamegraph^f_N$). Furthermore, the modification does not yield any new meaningful \Penv strategies. This is because all the newly added outgoing edges are to the same vertex $s' \in \overline{F}(s,u)$, that is already the successor of one $\Venv$ node $s^u_i$ in $\gamegraph^f_N$. So every winning \Penv strategy in $\gamegraph^f_{N+1}$ can be mimicked by a winning \Penv strategy in $\gamegraph^f_{N}$.

Before we start to prove~\cref{thm:dualmonotonicitya} formally, we cite a result from~\cite{SaglamS23} characterising strategies in fair Büchi games.
\begin{lemma}[\citet{SaglamS23}]~\label{lemma:almost-positional}
In fair Büchi games, \Psys has positional winning strategies, and \Penv has almost positional winning strategies. That is, let $\ltup{\gamegraph^f, \Inf(\XX)}$ be a fair Büchi game; if $v \in \Wenv(\gamegraph^f, \Inf(\XX))$, then there exists a \Penv winning strategy $\sigma$ s.t. from each $v \in \Venv$:
\begin{itemize}
\item If $v \not \in V^f$, then $\sigma$ is positional from $v$, i.e. $\sigma(h \cdot v) = w$ for every $h \in V^*$. 
\item If $v \in V^f$, then there is a successor $w' \in E(v)$ and a set $W = E^f(v) \cup \{w'\}$ such that in every play, when $v$ is visited for the first time, $\sigma$ takes the edge $(v, w')$, and when it is visited later, it cycles the edges $\{(v, w) \mid w \in W\}$.
\end{itemize}
\end{lemma}

\begin{proof}[Proof of \cref{thm:dualmonotonicitya}] 

We first let $G^f_{N+1}$ be the graph $G^f_N$ after a modification of type~\ref{item:OAshrinks} and show that $\Wsys(\mathcal{G}_{N+1}) \supseteq \Wsys(\mathcal{G}_N)  \cap V_{N+1}$.
This is the easier case, since the \Penv node $s^u$ loses a (non-fair) successor. This causes a decrease in the set of \Penv strategies. 
Formally, let $v \in \Wsys(\mathcal{G}_N) \cap V^{N+1}$ and $\iota$ be a \Psys winning strategy from $v$ in $\mathcal{G}_N$. As $\iota$ is winning against all \Penv strategies in $\mathcal{G}_N$, it is still winning against all \Penv strategies in $\mathcal{G}_{N+1}$; that is, $\iota$ is still a \Psys winning strategy from $v$ in $\mathcal{G}_{N+1}$, and therefore $v \in \Wsys(\mathcal{G}_{N+1})$.

Next we let $G^f_{N+1}$ be the graph $G^f_N$ after a modification of type~\ref{item:UAexpands} and show that $\Wsys(\mathcal{G}_{N+1}) \supseteq \Wsys(\mathcal{G}_N)$ (note that in this case, $V_N = V_{N+1}$ as the modification does not remove any node).

In $G^f_{N+1}$, all $s^u_i$ in $i \in [0, m^{s,u}]$ have a (fair) outgoing edge to $s'$. Since $s' \in \overline{F}^N(s, u)$, there exists a $j \in [1, m^{s,u}]$ such that the (fair) edge $(s^u_j, s')$ is in $G^f_N$. Let $v \in \Wsys(\mathcal{G}_N)$ and $\iota$ be a positional $\Vsys$ winning strategy from $v$ in $\mathcal{G}_N$ (\cref{lemma:almost-positional}).

Assume $v \in \Wenv(\mathcal{G}_{N+1})$ towards a contradiction. Then by~\cref{lemma:almost-positional}, there exists an almost positional \Penv winning strategy $\sigma$ that wins $v$ in $\mathcal{G}_{N+1}$.
As $s^u \not \in V^f$, $\sigma$ is positional on $s^u$. Let $s^u_i$ be the successor $\sigma$ assigns to $s^u$. 
Then the $\iota,\sigma$-play $\xi$ in $G^f_{N+1}$ that starts at $v$ and is compliant with both strategies $\iota$ and $\sigma$ is winning for \Penv. 
Since $\xi$ is \Penv winning, it is fair.
Then, $\xi$ either (i) sees $s^u$ (and therefore $s^u_i$) only finitely often, or (ii) it sees $s^u$ (and therefore $s^u_i$) infinitely often and 
cycles through the outgoing edges $(s^u_i, w)$ of $s^u_i$ infinitely often.

(i) Assume $\xi$ sees $s^u_i$ only finitely often. Then a tail (i.e. suffix) $\xi'$ of $\xi$ that does not see $s^u_i$ is a $\iota$-play in $G^f_N$, and in $G^f_{N+1}$; and therefore $\xi$ is \Psys winning. Note that this is because every node that can be reached from $s^u$ by a $\iota,\sigma$-play in $G^f_{N+1}$, can also be reached by a $\iota,\sigma'$-play in $G^f_N$ for a \Penv strategy $\sigma'$ in $G^f_N$. In particular, $s'$ can be reached from $s^u$ in $G^f_N$ via a \Penv strategy that takes the edges $s^u \to s^u_j \to s'$. Thus, we arrive at a contradiction.

(ii) Assume $\xi$ sees $s^u_i$ infinitely often.
Using $\sigma$, we will construct a \Penv strategy 
$\sigma'$ in $\mathcal{G}_N$ such that the $\iota,\sigma'$-play $\xi'$ that starts at $v$ in $G^f_N$ is winning for \Penv and therefore a contradiction. Let $\sigma'$ be a \Penv strategy in $G^f_N$ that is identical to $\sigma$ on every node in $V$ except $s^u$ and its successors, that is,  $\sigma'$ is almost positional with the same rules as $\sigma$ on every other node. On $s^u$, let $\sigma'$ to alternate between edges $(s^u, s^u_i)$ and $(s^u, s^u_j)$, and on $s^u_i$ and $s^u_j$ let $\sigma'$ to cycle through all the successors. Then since $s^u$ is visited infinitely often in $\xi'$, $s^u_i$ and $s^u_j$ are seen infinitely often as well; and consequently, all their successors. Therefore, all nodes in $\underline{F}^N(s,u) \cup \{s'\} \cup \{\overline{s}^{s,u}_i\}$ are seen infinitely often as well. Notice that these are exactly the successors of $s^u_i$ that are visited infinitely often in $\xi$ as well, since in $G^f_{N+1}$ the successors of $s^u_i$ are $\underline{F}^{N+1}(s, u) \cup \{\overline{s}^{s,u}_i\}$ where $\underline{F}^{N+1}(s, u) = \underline{F}^{N}(s, u) \cup \{s'\}$.
As $\iota$ is positional and $\sigma$ and $\sigma'$ are identical outside of $s^u$; $\xi$ and $\xi'$ see exactly the same nodes infinitely often. Therefore, $\xi'$ is \Psys winning. This gives us the desired contradiction and concludes the proof.
\end{proof}

%% file: appendix/app-sec5.tex
\section{Supplementary Material for Sec. 5}

In this section, we will prove~\cref{thm:cofaircoBuechiPM}---the main contribution of this paper. 
To achieve this, we need to go through intermediate steps. 
In the first subsection, we will introduce progress measures valid for (non-fair) coBüchi games. In the second subsection, we will give a gadget-based reduction from cofair coBüchi games to (non-fair) coBüchi games. Bringing these two steps together, in the third subsection, we will prove~\cref{thm:cofaircoBuechiPM}.

\subsection{The coBüchi Progress Measure}\label{subsec:cobuechi-pm}
\input{sections/pm.tex}

\subsection{The Gadget-Based Reduction}

In this section, we present 
\textit{gadgets} that turn a cofair coB\"{u}chi game into a regular coBüchi game. Our construction is inspired by the gadgets represented by~\citet{HausmannPSS24} that reduce a fair/cofair parity game into a regular parity game. 
If we directly use the gadgets given there to turn our cofair coB\"{u}chi game (i.e. a $\{2, 3\}$-color parity game) to a regular parity game, we  end up with a $\{1,2,3,4,5\}$-color parity game. However, this will introduce computational burden.
Therefore, inspired by these gadgets, we present simplified gadgets that turn a cofair coBüchi game into a (non-fair) coB\"uchi game. 

\begin{figure}[htbp]
  \centering
  \begin{minipage}{0.45\linewidth}
    \centering
    \begin{tikzpicture}[auto,node distance=1.5cm,semithick]
      \node[draw, circle, minimum size=0.8cm, inner sep=0pt] (q) {$v$}; 
      \node[draw, rectangle, minimum width=0.65cm, minimum height=0.65cm, inner sep=0pt]
        (ql) [below left=2em of q] {$v_l$};
      \node[draw, circle, minimum size=0.8cm, inner sep=0pt] 
      (qr) [below right=2em of q] {$v_r$};
      \node[draw, circle, minimum size=0.6cm, inner sep=0pt] at (qr.center) {};
      \node (eleft)  [below=1.5em of ql] {$E^f(v)$};
      \node (eright) [below=1.5em of qr] {$E(v)$};
      \draw[->] (q)  -- (ql);
      \draw[->] (q)  -- (qr);
      \draw[->] (ql) -- (eleft);
      \draw[->] (qr) -- (eright);
    \end{tikzpicture}
    \par\smallskip
    
  \end{minipage}\hfill
  \begin{minipage}{0.45\linewidth}
    \centering
    \begin{tikzpicture}[auto,node distance=1.5cm,semithick]
      \node[draw, circle, minimum size=0.8cm, inner sep=0pt] (q) {$v$}; 
      \node[draw, circle, minimum size=0.6cm, inner sep=0pt] at (q.center) {};
      \node (e)  [below=1.5em of q] {$E(v)$};
      \draw[->] (q)  -- (e);
    \end{tikzpicture}
    \par\smallskip 
  \end{minipage}
  \addtocounter{figure}{-4}
  \caption{Simplified cofair coBüchi gadgets for $v \in V^f \setminus \XX$ (left) and $v \in V^f \cap \XX$ (right). Doubly encircled nodes denote the coBüchi nodes.}
  \addtocounter{figure}{3}
\end{figure}
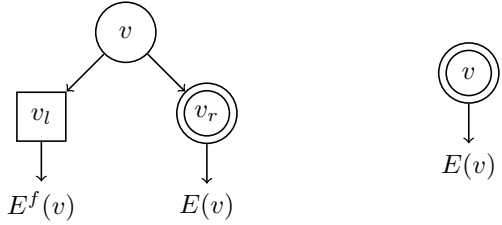

Given the cofair coBüchi game $\ltup{\gamegraph^{cf}, \overline{\Inf}(\XX)}$, the gadget game $\ltup{\gamegraph', \overline{\Inf}(\XX')}$ with $\gamegraph' = \tup{V', \Vsys', \Venv', E'}$ is obtained by replacing each fair node $v$ in $\gamegraph^{cf}$ with one of the gadgets given in~\cref{fig:gadgets}. If $v$ is a coBüchi node, it is replaced with the right gadget; all its successors are treated as regular successors, i.e., the fair outgoing edges of $v$ are turned into regular edges. If $v$ is not a coBüchi node, it is replaced with the left gadget and its incoming edges are redirected to the topmost gadget node (with the same name $v$), which has two children: the left child $v_l \in \Venv' \setminus \XX'$, and the right child $v_r \in \Vsys' \cap \XX'$. The left child has fair successors of $v$, and the right child has all successors of $v$. %
 
Next, we claim that the gadget game preserves the winning regions of the original game.

 \begin{restatable}{thm}{gadgetgame}\label{thm:gadgetgame}
Let $\ltup{\gamegraph^{cf}, \overline{\Inf}(\XX)}$ be a cofair coBüchi game with $\gamegraph^{cf} = (\gamegraph, E^f)$ and $\gamegraph = \tup{V, \Vsys, \Venv, E}$. There exists a coBüchi game $\ltup{\gamegraph', \overline{\Inf}(\XX')}$ with $\gamegraph' = \tup{V', \Vsys', \Venv', E'}$ with $V \subseteq V'$, where $|V'| = |V| + 2\cdot |V^f \setminus \XX|$ and $|\XX'| = |\XX| + |V^f \setminus \XX|$. Furthermore, for $j \in \{0,1\}$,  $\Win_j(\gamegraph^{cf}, \overline{\Inf}(\XX)) \cap V = \Win_j(\gamegraph', \overline{\Inf}(\XX')) \cap V$. 
\end{restatable}

\input{appendix/app-gadgetgame}

\subsection{The cofair coBüchi PM (Proof of~\cref{thm:cofaircoBuechiPM})}

Bringing~\cref{thm:gadgetgame} and the valid coBüchi PM together gives us the valid cofair coBüchi PM in~\cref{thm:cofaircoBuechiPM}. 
\cofaircoBuechiPM*

We obtain the cofair coBüchi progress measure in~\cref{eq:cofaircobuechi-pm,eq:cofaircobuechi-prf} by \emph{restricting the coBüchi progress measure} given in~\cref{eq:cobuechi-pm,eq:cobuechi-pr} applied to the gadget game, \emph{to the nodes $V$} of the original cofair coBüchi game. 

Observe that the cofair coBüchi game graph $\gamegraph^{cf}$ \emph{embeds in} the gadget game graph $\gamegraph'$ (\cref{thm:gadgetgame}). While we apply the coBüchi progress measure in~\cref{eq:cobuechi-pm} to $V'$, we focus on the way that the PM changes for $v \in V$  w.r.t. the PM of its regular and fair successors.

Take a node $v \in V^f \setminus \XX$ in $G^{cf}$. Then in $G'$, $v$ is the root node of the gadget in the left figure of~\cref{fig:gadgets}. 
Then, $\Progc(v) = \min \{\Progc(v_l), \Progc(v_r) + 1\}$ (as $v_r\in \XX'$) where $\Progc(v_l) = \min_{w \in E^f(v)} \rho(w)$ and $\Progc(v_r) = \max_{w \in E(v)} \rho(w)$. Therefore, when we restrict $\Progc$ to $v$ and its successors in $V$, we obtain, as in~\cref{eq:cofaircobuechi-prf},
$$\Prog(v) = \min \begin{cases}  
        \max_{(v,w) \in E^f} \rho(w),  \\
        \min_{(v,w) \in E}\rho(w) + 1.  
    \end{cases}$$
Similarly, for a node $v$ in $V^f \cap \XX$,  simply
$\Prog(v) = \Progc(v)$ since in the gadget game, $v$ is treated as a non-fair node. For non-fair nodes in $G^{cf}$  $\Prog(v) = \Progc(v)$ as well since their outgoing edges are the same in $G^{cf}$ and $G'$.

\paragraph{Enhanced PM range for (the dual of) the abstract game} Note that $|\XX'| = |\XX| + |V^f \setminus \XX|$ due to the construction of the gadget game (\cref{thm:gadgetgame}). In the abstract fair Büchi game the size of $|V^f|$ depends on the difference in the sizes of the under- and over-approximation of each $(s, u)$ pair (recall the abstract fair game graph construction from Sec. 3), and is upper-bounded by $|S|^2\cdot |\mathcal{U}|$ as $\overline{F}(s, u)$ can be the entire $S$ for each $u \in \mathcal{U}$. Clearly, $V^f$ is the same in the dual game. Therefore, the range of the cofair coBüchi PM is in the worst case $ L = |S|^2\cdot |\mathcal{U}|$. In order to optimize the number of iterations in the lifting algorithm, we observe that for the dual of the abstact fair Büchi game, we can restrict the cofair coBüchi PM to the range $[1, |S| + |\XX|] \cup \{\top\}$ (where $\XX\subseteq S$ is the set of \emph{abstract Büchi states}). 

Intuitively this is because due to the DAG-like structure of the subgame in~\cref{fig:gameabstraction}; \Penv cannot ensure visiting two different successors $s^u_i$ and $s^u_j$ of a node $s^u$. 
Recall that in the dual game, the ownership of the nodes is swapped, so $s^u \in \Vsys$. Now imagine replacing the $V^f$ nodes with the left gadget in~\cref{fig:gadgets} to obtain the gadget game of the dual of the abstract game. Since the gadget game is a coBüchi game without fairness, both players have positional strategies. Fix a positional winning \Psys strategy $\iota$ from all $v \in \Wsys$. In particular for each $s^u$, $\iota(s^u) = s^u_i$ for some $i$ and $\iota(s^u_i) = (s^u_i)_p$ for some $p \in \{r, l\}$. So each cycle that passes through $s^u$, can pass through only one of $s^u_i$.

Similar to the explanation of the range for coBüchi PMs, $\rho(v) > |S| + |\XX|$ implies \Penv has a strategy to ensure $|S|+|\XX| + 1$ coBüchi nodes are seen. This implies either (i) more than one successor $s^u_i$ of some $s^u$ are seen, or (ii) a cycle containing a coBüchi node is seen. As argued before, case (i) is not possible. Therefore, if $\rho(v)$ is lifted above $|S|+|\XX|$, \Penv has a winning strategy from $v$.

%% file: sections/pm.tex
First presented by~\citet{KK91}, progress measures became more known for Jurdzinski's presentation~\cite{Jurdzinskismall00} for parity games, under the name ``small progress measures.'' We adopt the terminology of the later work, and restrict it to coBüchi games---as coBüchi games are 2-color ($\{2, 3\}$) parity games.\footnote{The parity progress measures in~\citet{Jurdzinskismall00} are restricted to colors $\{2,3\}$ and yield different results than ours. In fact, it yields PMs that assign vertices values in $\mathbb{N}^2$. However, it is not difficult to observe that one value of each pair is redundant, and they are equivalent to the one-digit coBüchi PMs we present in this work.}%

A progress measure $\rho: V \to [0, L] \cup \{\top\}$ 
is \emph{valid for coBüchi} if it satisfies~\cref{eq:cobuechi-pm} and~\eqref{eq:cobuechi-pr} where the range is given by $L = |\XX|$, with $\XX$ denoting the set of coBüchi vertices. The induced lifting function~\cref{eq:cobuechi-lifting} is the same as~\cref{eq:cofaircoBuechi-lifting} in the main paper---their difference comes from different definitions of $\Prog$.

\begin{equation}\label{eq:cobuechi-pm}
\rho(v) \geq \begin{cases}
    \Progc(v) + 1 \quad &\text{if} \quad v\in \XX,\\
    \Progc(v) \quad &\text{if}\quad v \not\in \XX,
\end{cases}
\end{equation}
\begin{equation}\label{eq:cobuechi-pr}
\text{where } \Progc(v) = \begin{cases}
    \min_{(v,w) \in E}\rho(w) \quad &\text{if} \quad v \in \Vsys,\\
    \max_{(v,w) \in E}\rho(w) \quad &\text{if} \quad v \in \Venv,
\end{cases}
\end{equation}
\begin{equation}\label{eq:cobuechi-lifting}
\textsf{Lift}(\rho, v) = \rho' \, \text{where} \,\,\rho'(w) = 
\begin{cases}
    \rho(w)  &\text{if}\,\quad w \neq v,\\     
    \Progc(v)+1 &\text{if}\, v = w \in \XX, \\ 
    \Progc(v) &\text{if}\, v = w \not \in \XX.
\end{cases}
\end{equation}

Knaster-Tarski Theorem\footnote{Using the facts that the set of Büchi progress measures of a game graph is a complete lattice where $\rho^0 : V \to \{0\}$ is the least element, and the lifting function is monotonic.}~\cite{Tarski95} reveals that by applying the lifting function to the smallest progress measure $\rho^0 : V \to 0$ iteratively, we reach a fixpoint. This fixpoint $\rho^*$, carries information of the winning regions. Specifically, $\Wsys(\gamegraph^{cf}, \overline{\Inf}(\XX)) = \{v \in V \mid \rho^*(v) \neq \top\}$, and consequently, $\Wenv(\gamegraph^{cf}, \overline{\Inf}(\XX)) = \{v \in V \mid \rho^*(v) = \top\}$.

\paragraph{The intuition behind the range of $\rho$} It is enough to restrict the range of a valid coBüchi PM to $L = |\XX|$ because if a node $v$ is lifted until it has reached $\rho(v) > |\XX|$, this means \Penv has a strategy $\sigma$ to ensure a coBüchi node is reached repeatedly from $v$, making all $\sigma$-plays, and $v$, \Penv winning. Recall that in coBüchi games both players have positional strategies. Observe that a node $v$ gains $\rho(v) = m$ if there exists a \Penv strategy $\sigma$ that in a play starting at $v$, it makes sure  $m$ coBüchi nodes are visited.\footnote{$\sigma(v) = w$ where $w$ is the successor with maximum $\rho(w)$.} Thus, if $\rho(v) > |\XX|$, this means \Penv has a strategy $\sigma$ to ensure $|\XX| + 1$ coBüchi nodes are visited. Therefore, $\sigma$ ensures visiting a cycle that contains at least one coBüchi node, and therefore visiting the coBüchi set infinitely often.

%% file: appendix/app-gadgetgame.tex
First, we formally define the gadget game $\ltup{\gamegraph', \XX'}$ where $\gamegraph' = \tup{V', \Vsys', \Venv', E'}$ is constructed using the gadgets in~\cref{fig:gadgets}:
\begin{align*}
&V' = V \,\cup\, \{\,v_l, v_r \mid v \in V^f \setminus \XX\,\} ,\\
&\XX' = \XX \,\cup\, \{\, v_r \mid v \in V^f \setminus \XX\,\},
\end{align*}
and for $v \in V$,
\begin{align*}
&E'(v) = \begin{cases} 
     \,\{v_l, v_r\} & \text{ if } v  \in V^f \setminus \XX,\\
    \,E(v)  &\text{ otherwise,} 
\end{cases}\\
&E'(v_l) = E^f(v), \quad \quad E'(v_r) = E(v).
\end{align*}
It follows directly from the construction that the bounds on $|V'|$ and $|\XX'|$ given in~\cref{thm:gadgetgame} hold. We will prove the main part of the theorem, i.e.,
$\Win_j(\gamegraph^{cf}, \overline{\Inf}(\XX)) \cap V =  \Win_j(\gamegraph', \overline{\Inf}(\XX')) \cap V$ for $j \in \{0,1\}$.

Before we start the proof, we recall that both players have positional winning strategies in coBüchi games, e.g., the gadget game. 

\begin{proof}
We show the equivalence in two parts: $ \textbf{(1) } \Wsys(\gamegraph', \overline{\Inf}(\XX'))  \cap V \subseteq \Wsys(\gamegraph^{cf}, \overline{\Inf}(\XX)) \cap V$, and then \textbf{(2)} $\Wenv(\gamegraph', \overline{\Inf}(\XX')) \cap V \subseteq\Wenv(\gamegraph^{cf}, \overline{\Inf}(\XX)) \cap V$. By the determinacy of cofair coBüchi games (i.e., $V = \Wsys \uplus \Wenv$), this will give us the desired equality. 

\smallskip
\textbf{(1)} $\Wsys(\gamegraph', \overline{\Inf}(\XX')) \cap V \subseteq \Wsys(\gamegraph^{cf}, \overline{\Inf}(\XX)) \cap V$:

Let $v \in \Wsys(\gamegraph', \overline{\Inf}(\XX'))$ and $\iota'$ be a positional \Psys winning strategy in $\ltup{\gamegraph', \overline{\Inf}(\XX)}$. Using $\iota'$, we will construct a \Psys winning strategy $\iota$ which wins $v$ in $\ltup{\gamegraph^{cf}, \overline{\Inf}(\XX)}$.

We construct $\iota$ as follows:
Whenever $w \not \in V^f$ or $w \in V^f \cap \XX$, $\iota$ is positional on $w$ with $\iota(w) = \iota'(w)$. Otherwise, that is when $w \in V^f \setminus \XX$: If $\iota'(w) = w_r$, then set $\iota$ to be positional again on $w$ with $\iota(w) = \iota'(w_r)$. If $\iota'(w) = w_l$, then whenever $w$ is visited, make $\iota$ to cycle through all fair successors of $w$, according to an arbitrary, fixed order on $E^f(w)$.

Now we show that $\iota$ is a winning \Psys strategy in $\ltup{\gamegraph^{cf}, \overline{\Inf}(\XX)}$ that wins $v$. For this, we need to introduce the concept of \emph{extension} of an $\iota$-play in $\ltup{\gamegraph^{cf}, \overline{\Inf}(\XX)}$ to the gadget game.

Take an $\iota$-play $\xi = w^1 w^2 \ldots$ in $\gamegraph^{cf}$ with $w^1 = v$. The \emph{extension} $\xi'$ of $\xi$ to $\gamegraph'$ is defined as follows: Initialize $\xi'$ with $\xi' = w^1$, and increment $j=1$ one-by-one. 
For $w^j \not \in V^f $ or $w^j \in V^f \cap \XX$, extend $\xi'$ with $w^{j+1}$. Otherwise, i.e., for $w^j \in V^f \setminus \XX$, let $\iota'(w^j) = w^j_p$ where $p \in \{l, r\}$, and extend $\xi'$ with $w^j_p w^{j+1}$.

 As $\xi$ is an $\iota$-play, whenever $\iota'(w^j) = w^j_l$, $w^{j+1} \in E^f(w^j)$ and whenever $\iota'(w^j) = w^j_r$, $w^{j+1} = \iota'(w^j_r)$. Therefore, $\xi'$ is an $\iota$-play in $\gamegraph'$. As it starts from $v$, $\xi'$ is also winning for \Psys in the gadget game. Furthermore, note that if we restrict the nodes in $\xi'$ to $V$, i.e., skip the gadget nodes $v_l$ and $v_r$, we get exactly $\xi$. We denote this phenomenon by $\xi'|_{V} = \xi$.

 Let $\tilde{\xi'}$ denote the tail of $\xi'$ consisting of nodes and edges seen infinitely often in $\xi'$. %
 Note that as $\xi'$ is \Psys winning, $\tilde{\xi'}$ %
  does not contain coBüchi nodes. 

 Using these observations, we will show that $\xi$ is \Psys winning. %
In order to show $\xi$ is \Psys winning we need to show two properties: (i) $\xi$ is fair, and (ii) $\tilde{\xi} = \tilde{\xi'}|_{V}$ does not contain any coBüchi nodes.

Property (ii) follows directly from $\tilde{\xi}$ not containing any coBüchi nodes. To see (i), observe that by construction of $\iota$, the only fair nodes that are not required to traverse their fair outgoing edges are coBüchi nodes and nodes $w$ with $\iota'(w) = w_r$. Recall that $w_r$ is a coBüchi node in the gadget game. As $\tilde{\xi'}$ does not contain any coBüchi nodes, it does not contain any nodes $w_r$, and by the definition of extension, this means $\tilde{\xi}$ does not contain any nodes $w$ with $\iota'(w) = w_r$.
Therefore, we conclude that $\tilde{\xi'}|_{V} = \tilde{\xi}$ does not contain any fair nodes that do not traverse all their fair outgoing edges.  

Thus, $\xi$ satisfies both (i) and (ii), and is therefore \Psys winning. As $\xi$ is an arbitrary $\iota$-play that starts at $v$ in $\ltup{\gamegraph', \overline{\Inf}(\XX)}$, this implies that $\iota$ is a winning strategy, and $v \in \Wsys(\gamegraph^{cf}, \overline{\Inf}(\XX))$.

This concludes the first part of the proof. We now proceed to the second part. 

\smallskip
\textbf{(2)} $\Wenv(\gamegraph', \overline{\Inf}(\XX')) \cap V \subseteq \Wenv(\gamegraph^{cf}, \overline{\Inf}(\XX)) \cap V$:

Let $v \in \Wenv(\gamegraph', \overline{\Inf}(\XX'))$ and $\sigma'$ be a positional \Penv winning strategy in  $\ltup{\gamegraph', \overline{\Inf}(\XX)}$. For every $w \in V^f \setminus \XX$, let $\sigma(w)$ denote $\sigma'(w_l)$. As before, using $\sigma'$, we will construct a \Penv winning strategy $\sigma$ in $\ltup{\gamegraph', \overline{\Inf}(\XX)}$ that wins $v$. The strategy $\sigma$ is set by simply assigning $\sigma(w)  = \sigma'(w)$ for each $w \in \Venv$. 

Again, before we show that $\sigma$ is winning for \Penv we need to introduce the concept of \emph{extension} of a $\sigma$-play in $\ltup{\gamegraph^{cf}, \overline{\Inf}(\XX)}$ to the gadget game. The definition of the extension is different from that in the first part of the proof.

Take a $\sigma$-play $\xi = w^1 w^2 \ldots$ in $\gamegraph^{cf}$ with $w^1 = v$. The \emph{extension} $\xi'$ of $\xi$ to $\gamegraph'$ is defined as follows: Initialize $\xi'$ with $\xi' = w^1$ and increment $j =1 $ one-by-one. For $w^j \not \in V^f $ or $w^j \in V^f \cap \XX$, extend $\xi'$ with $w^{j+1}$. For $w^j \in V^f \setminus \XX$, if $w^{j+1} = \sigma(w^j)$ extend $\xi'$ with $w^j_l w^{j+1}$, otherwise extend $\xi'$ with $w^j_r w^{j+1}$.

Similar to the first direction, we note that $\xi'$ is a $\sigma$-play in the gadget game starting at $v$, and is therefore winning for \Penv. Furthermore, $\xi'|_{V} = \xi$. We define $\tilde{\xi'}$ as before and let $\gamegraph'|_{\tilde{\xi'}}$ denote the subgraph of $\gamegraph'$ that consists of nodes and edges taken in $\tilde{\xi'}$. Note that as $\tilde{\xi'}$ is a $\sigma$-play and $\sigma$ is positional, every node $\Venv$ in $\gamegraph'|_{\tilde{\xi'}}$  has exactly one outgoing edge. Therefore, every play in $\gamegraph'|_{\tilde{\xi'}}$ is realizable by a \Psys strategy, i.e., for every play $\tau$ in $\gamegraph'|_{\tilde{\xi'}}$, there exists a \Psys strategy $\iota'$ such that $\tau$ is a tail of the $\iota', \sigma'$-play that starts at $v$. Therefore, every cycle in $\gamegraph'|_{\tilde{\xi'}}$ is \Penv-winning.  This means every simple cycle carries a coBüchi node. 

Similar to the first part, using these observations, we will show that $\xi$ is \Penv winning. In order to show that $\xi$ is \Penv winning we need to show one of the two holds: (i) $\xi$ is unfair, or (ii) $\tilde{\xi} = \tilde{\xi'}|_{V}$ visits a coBüchi node. 

We will show that $\neg (i ) \Rightarrow (ii)$ holds, i.e., if $\xi$ is fair, then $\tilde{\xi}$ visits a coBüchi node. 
Assume $\xi$ is fair. Then every node $w \in V^f \cap \tilde{\xi}$ traverses all their fair outgoing edges infinitely often in $\tilde{\xi}$. Specifically, the edge $(w, \sigma(w))$ is traversed infinitely often as $\sigma(w) \in E^f(w)$. Then, the branch $w \to w_l \to \sigma(w)$ is traversed infinitely often in $\tilde{\xi'}$, and therefore exists in $\gamegraph'|_{\tilde{\xi'}}$. In other words, all fair nodes in $\gamegraph'|_{\tilde{\xi'}}$ have their left branches in the subgraph. We proceed to remove all right branches from $\gamegraph'|_{\tilde{\xi'}}$ and obtain yet another subgraph, and call it $H$. As all fair nodes had their left branches as well in  $\gamegraph'|_{\tilde{\xi'}}$, $H$ contains no dead-ends. As $H$ is a subgraph of $\gamegraph'|_{\tilde{\xi'}}$ without a dead-end, it contains a cycle, and consequently contains a coBüchi node. Due to the absence of right branches (therefore, the nodes $w_r$) in $H$, the only coBüchi nodes come from nodes in $V$. So $H$, and therefore $\tilde{\xi'}$, contain a coBüchi node $u \in V$. Then, $u \in \tilde{\xi} = \tilde{\xi'}|_{V}$. This concludes the proof of the implication.

As $\xi$ is an arbitrary $\sigma$-play that starts at $v$ in $\ltup{\gamegraph', \overline{\Inf}(\XX)}$, this implies that $\sigma$ is a winning strategy, and $v \in \Wenv(\gamegraph^{cf}, \overline{\Inf}(\XX))$. 

This concludes the second part of the proof, and therefore the proof of~\cref{thm:gadgetgame}.
\end{proof}

%% file: appendix/app-sec6.tex
\section{Supplementary Material for Sec. 6}

In this section we will detail the ingredients we need to ensure the end-to-end efficiency of~\cref{alg:main_fixed}. 
In particular, we will show how to seamlessly integrate a symbolic fixpoint solver with our incremental algorithm. For this, we will show how to extract $\rho^*$ from the $\texttt{initialise}$ function in~\cref{alg:main_fixed} (line 1), which solves the fixpoint given in~\cref{lemma:cofaircoBuechifp}. The main theorem of this section is the following.
\begin{restatable}{thm}{fixpointpmequivalence}\label{thm:fixpoint-pm-equivalence}
    Let $\mathcal{G}$ be a fair Büchi game with the vertex set $V$. Then for every $v \in V$, 
    $\rho^{\psi}(v) = \rho^*(v)$, where $\rho^{\psi}$ is obtained from iterations of \cref{eq:fixpoint-cofaircoBuechi} and $\rho^*$ is the l.f.p. of the lifting function~\cref{eq:cofaircoBuechi-lifting} applied on the dual game $\overline{\mathcal{G}}$.
\end{restatable}

In the first subsection, we will first briefly introduce fixpoint calculations and the notation we use in~\cref{eq:fixpoint-cofaircoBuechi}. In the second subsection, we will show how to extract a PM $\rho^{\psi}$ from the calculation of~\cref{eq:fixpoint-cofaircoBuechi}. In the last subsection, we will prove~\cref{thm:fixpoint-pm-equivalence}. While the first and the second subsections are known in the literature, the last subsection is our contribution that allows for the integration of symbolic fixpoint solvers with our incremental algorithm~\cref{alg:main_fixed}.

\subsection{Fixpoint Algorithms in the $ \mu $-calculus} 

The $\mu $-calculus offers a succinct representation of symbolic algorithms (i.e., algorithms manipulating sets of vertices instead of individual vertices) over a game graph. 
In this section, we omit the general preliminaries of $\mu$-calculus and instead focus on the evaluation procedure for the two-nested formula defined in~\cref{eq:Z}. We refer to~\citet{Kozen:muCalculus} for further details regarding the $\mu$-calculus.
\begin{equation}\label{eq:Z}
Z=\mu Y.~\nu X.~\phi(X,Y)
\end{equation} %
Here, $X$ and $Y$ are subsets of $V$, $\mu$ denotes the least fixpoint, $\nu$ the greatest fixpoint, and $\phi$ is a formula composed of the \emph{monotone set transformers} introduced in~\cref{eq:transformers1,eq:transformers2}.

To compute $Z$, we start by initializing $X$ and $Y$: $X^0 := V$ (as $X$ is under $\nu$) and $Y^0 := \emptyset$ (as $Y$ is under $\mu$).  We first keep $Y$ at its initial value and iteratively compute $X^k=\phi(X^{k-1},Y^0) $ until $X^{k+1}=X^k$. 
At this point $X$ reaches a fixpoint; we denote this value of $X$ by $X^\infty$ and assign this value to $Y$, i.e., we have $Y^1 = X^\infty$.  We re-initialize $X^0:=\emptyset$, and re-evaluate $X^k=\phi(X^{k-1},Y^1) $ with the new value of $Y$. 
When $Y$ reaches a fixpoint, that is, $Y^\infty = Y^{k+1} = Y^k$, the calculation terminates and the evaluation of the formula equals $Z = Y^\infty$.

In order to retain all intermediate values of $X$, we use $X^{l,k}$ to denote the set which is computed in the $k$-th iteration over $X$ during the computation of $Y^l=X^{l,\infty}$.

\subsection{Extracting the Mapping Function} 
It is well-established that nested $\mu$-calculus formulas induce a mapping from the set of vertices to the natural numbers. In particular, a two-nested $\mu$-calculus formula $Z$, e.g.~\cref{eq:Z}, yields the mapping function $f^Z$ where for each $v \in Z$, $f^Z(v)$ is the unique integer $i$ for which $v \in Y^{i+1} \setminus Y^{i}$.\footnote{Note that typically $f^Z(v)$ is taken to be $i$ for which $v \in Y^{i} \setminus Y^{i-1}$, but we utilize this shifted indexing to simplify the notation in our proofs.}

\subsection{Obtaining $\psi$ and $\rho^\psi$} 

It is shown by~\citet{banerjee2023fast} that the following 2-nested $\mu$-calculus formula $\Psi$ calculates the \Psys winning region in a fair Büchi game $\mathcal{G} = \ltup{\gamegraph^f, \Inf(\XX)}$ with $\gamegraph^f = (\gamegraph, E^f)$ and $\gamegraph = \tup{V, \Vsys, \Venv, E}$:
\begin{align}\label{eq:fp-fairBuechi}
\Psi = 
\nu Y.\, \mu X. \, &( \XX \,\cap \,\Cpresys(Y) ) \, \cup \, \Cpresys(X) \,\cup\, \\ &(\Lpreexists(X) \,\cap\,\mathsf{Pre}^{\forall}_1(Y)).\notag
\end{align}
Here, $\XX$ is the set of Büchi nodes, and\begin{align}\label{eq:transformers1}
    &\mathsf{Pre}_1^{\forall}(H)= \{ v \in \Venv \, \mid \, E(v) \subseteq H\},\\ \notag
    &\mathsf{Pre}_0^{\exists}(H)= \{ v \in \Vsys \, \mid \, E(v) \cap H \neq \emptyset\},\\ \notag
    &\Cpresys(H) =  \mathsf{Pre}_1^{\forall}(H) \cup \mathsf{Pre}_0^{\exists}(H) \text{,}\\ \notag
    &\Lpreexists(H) = \{ v \in V^f \, \mid \, E^f(v) \cap H \neq \emptyset\}. 
\end{align}
As fair Büchi games are determined ($V = \Wsys(\mathcal{G}) \uplus \Wenv(\mathcal{G})$), we can negate~\cref{eq:fp-fairBuechi} and obtain the following $\mu$-calculus formula $\psi$ that calculates $\Wenv(\mathcal{G})$: 
\begin{align}%
\psi =  \mu Y.\, \nu X. \, &( (\neg \XX) \,\cup \,\Cpre(Y) ) \, \cap \, \Cpre(X) \, \cap \,\tag{6} \\ &( \Lpre(X) \, \cup \, \Preexists(Y)\, \cup \, (\neg V^f)),\notag
\end{align}
where $\neg H$ stands for $V \setminus H$ for a subset $H \subseteq V$, and
\begin{align}\label{eq:transformers2}
    &\Preforall(H) = \{ v \in \Vsys \, \mid \, E(v) \subseteq H\},\\ \notag
    & \Preexists(H) = \{ v \in \Venv \, \mid \, E(v) \cap H \neq \emptyset\},\\ \notag
    &\Cpre(H) =  \Preforall(H) \cup \Preexists(H) \text{,}\\ \notag
    &\Lpre(H) = \{ v \in V^f \, \mid \, E^f(v) \subseteq H\}.
\end{align}
Note that~\cref{eq:fixpoint-cofaircoBuechi} is discussed in~\cref{lemma:cofaircoBuechifp}. It is straightforward to verify~\cref{eq:fp-fairBuechi,eq:fixpoint-cofaircoBuechi} are negations of each other. As $\Wenv(\mathcal{G}) = \Wsys(\overline{\mathcal{G}})$, $\psi$ also calculates the \Psys winning region in the dual cofair coBüchi game $\overline{\mathcal{G}}$. %

Finally, we define $\rho^\psi: V \to \mathbb{N} \cup \{\top\}$ such that for $v \in \psi$, $\rho^\psi(v) = f^\psi(v)$, and for $v \in \overline{\psi}$, $\rho^{\psi}(v) = \top$.

\subsection{Proof of~\cref{thm:fixpoint-pm-equivalence}}

Recall that $\rho^*$ is the least fixpoint of the lifting function in~\cref{eq:cofaircoBuechi-lifting} on the dual cofair coBüchi game $\overline{\mathcal{G}} = \ltup{\gamegraph^{cf}, \overline{\Inf}(\XX)}$ of $\mathcal{G} = \ltup{\gamegraph^f, \Inf(\XX)}$. The only difference between the game graphs $\gamegraph^f$ and $\gamegraph^{cf}$ is that the ownership of the nodes is swapped. Therefore, we obtain the following rules by swapping the ownership of the nodes in~\cref{eq:cofaircobuechi-prf,eq:cofaircoBuechi-lifting}:
\begin{equation}\label{eq:fairbuechi-pr}
\Prog(v) = \begin{cases}
    \min \begin{cases}  
        \max_{(v,w) \in E^f} \rho(w)  \\
        \min_{(v,w) \in E}\rho(w) + 1  
    \end{cases} &\text{if} \, v \in V^f \setminus \XX, 
    \\
    \min_{(v,w) \in E}\rho(w)  &\hspace{-0.84cm}\text{if} \, v \in \Venv  \setminus (V^f \setminus \XX), \\
    \max_{(v,w) \in E}\rho(w)  &\hspace{-0.84cm}\text{if}  \,v \in \Vsys. \\
\end{cases}
\end{equation}

\begin{equation}\label{eq:fairBuechi-lifting}
\textsf{Lift}(\rho, v) = \rho' \,\text{where} \,\,\rho'(w) = 
\begin{cases}
    \rho(w)  &\text{if}\,\quad w \neq v,\\     
    \Prog(v)+1 &\text{if}\, v = w \in \XX, \\ 
    \Prog(v) &\text{if}\, v = w \not \in \XX.
\end{cases}
\end{equation}

The fixpoint formula in ~\cref{eq:fixpoint-cofaircoBuechi} calculates the winning region of the fair Büchi game $\Wenv$. Here, we provided a lifting algorithm that works directly on the fair Büchi game (and not the dual), given in~\cref{eq:fairbuechi-pr,eq:fairBuechi-lifting}, to simplify the comparison of the mappings $\rho^*$ and $\rho^\psi$. Furthermore, we actually use this version of the lifting algorithm in the $\texttt{initialise}$ function in~\cref{alg:main_fixed} in order to bypass the construction of the dual game. Before we start the proof, we will introduce one last definition. 

\begin{definition}[(Full-)$\Venv$-subgraph]
    Let $\ltup{\gamegraph^f, \Inf(\XX)}$ be a fair Büchi game with $\gamegraph^f = (\gamegraph, E^f)$ where $\gamegraph = \tup{V, \Vsys, \Venv, E}$. A subgraph $\subgraph = \tup{V', \Vsys', \Venv', E'}$ of $\gamegraph$ is called a \textit{full-$\Venv$-subgraph} if it contains all successors of $\Vsys$ nodes, at least one successor of each $\Venv$ node, and all fair successors of each fair non-Büchi node. 
    That is, in $\subgraph$, 
    \begin{align}\label{eq:full-subgraph}
        &\text{for } \quad  v \in V' \cap \Vsys, \quad &&E'(v) = E(v),\notag \\
        &\text{for } \quad  v \in V' \cap (\Venv \cup \XX), \quad &&|E'(v)| \geq 1,\\
        &\text{for }\quad  v \in V' \cap V^f \cap \neg \XX, \quad &&E'(v) = E^f(v).\notag
    \end{align}
    
    A subgraph is called a \textit{$\Venv$-subgraph} if it is a \textit{full-$\Venv$-subgraph} with the exception that nodes are allowed to be dead-ends. That is, each $v$ in a $\Venv$-subgraph $\subgraph$ is either a dead-end, or satisfies~\cref{eq:full-subgraph}.
\end{definition}

\fixpointpmequivalence*

For $v \in \Wenv(\overline{\mathcal{G}}) = \Wsys(\mathcal{G})$,  $\rho^*(v) = \rho^{\psi}(v) = \top$. Therefore, we only need to show $\rho^*(v) = \rho^{\psi}(v) $ for $v \in \Wsys(\overline{\mathcal{G}}) = \Wenv(\mathcal{G})$.

\begin{proof}
We proceed by induction on the natural numbers assigned to the nodes in $\Wenv(\mathcal{G})$ by $\rho^*$ and $\rho^\psi$. 

\smallskip
\emph{Base case:} $\rho^*(v) = 0 \Leftrightarrow \rho^{\psi}(v) = 0$.

\textit{Proof of the base case. } Observe that $\rho^\psi(v) = 0 $ if and only if $v \in Y^1$ where
\begin{equation}\label{eq:pm-zero} Y^1 \,= \, \nu X. \, \neg \XX  \, \cap \, \Cpre(X) \, \cap \, ( \Lpre(X) \, \cup \, \neg V^f),
\end{equation}
because $Y^0 = \emptyset$ and  thus $\Cpre(Y^0) =\Preexists(Y^0) = \emptyset$.

We claim that $Y^1$ is the maximal subset of $\neg \XX$ that accepts a full-$\Venv$-subgraph in $\gamegraph$.
$Y^1$ accepts a full-$\Venv$-subgraph because all $v \in Y^1$ belong to $\Cpre(Y^1)$ (recall that in the last iteration of $X^0$, $X^{0, j} = X^{0,j+1} = Y^1$) and thus for every $v \in \Venv \cap Y^1$, there exists a successor in $Y^1$, and for every $v \in \Vsys \cap Y^1$, all of the successors are in $Y^1$. Furthermore, every $v \in V^f \cap Y^1$ belongs to $\Lpre(Y^1)$, thus all of their fair successors are in $Y^1$. Note that $Y^1 \subseteq \neg \XX$. 

Next we show that there is no subset $H$ of $\neg \XX$ that accepts a full-$\Venv$-subgraph such that $H \not \subseteq Y^1$. %
Observe that the computation of the set $Y^1 = X^\infty$ starts from $X^0 = V$ followed by $X^1 = \neg \XX$ and can only get smaller in the next iterations. 
For $w \not \in X^2$, there is a \emph{one-step way} to reach a node in $\XX$. Here, by a \emph{one-step way}, we mean for $w \in \Vsys$, there is a successor of $w$ in $\XX$, and for $w \in \Venv$, all successors of $w$ are in $\XX$. Therefore, unless $X^1 = X^2$, the set $X^1$ does not accept a full-$\Venv$-subgraph. Similarly, for $w \not \in X^3$, there is a \emph{two-step way} to reach a node in $\XX$. Therefore, unless $X^2 = X^3$, the set $X^2$ does not accept a full-$\Venv$-subgraph. 
It follows by induction that $H \subseteq Y^1$.

In addition, it follows from~\cref{eq:fairbuechi-pr,eq:fairBuechi-lifting} that $\rho^*(v) = 0$ if and only if $v \not \in \XX$ and,
\begin{align*}
    &\text{if $v \in V^f$,} &&\text{all fair successors $w$ of $v$ have $\rho^*(w) = 0$, }\\
    &\text{if $v \in \Venv \setminus V^f $,} &&\text{there exists a $w \in E(v)$ with $\rho^*(w) = 0$, }\\
    &\text{if $v \in \Vsys$,} &&\text{all successors $w$ of $v$ have $\rho^*(w) = 0$. }
\end{align*}
It is easy to see that this exactly corresponds to the set $\{v \in V \mid \rho^*(v) = 0\}$ being the maximal subset of $\neg \XX$ that accepts a full-$\Venv$-subgraph in $\gamegraph$. 
 Note that this in particular implies that $\rho^*(v) = 0$ iff $v\in Y^1\setminus Y^0$ as $Y^0=\emptyset$.
This completes the proof of the base case.

\smallskip
\emph{Induction Hypothesis: } For every $k < m \neq 0$, $\rho^{\psi}(v) = k \Leftrightarrow \rho^*(v) = k$. In other words,  
 $\rho^*(v) = k$ iff $v\in Y^{k+1}\setminus Y^k$.
 
\smallskip
\emph{Inductive Step: } We will show $ \rho^{\psi}(v) = m \Leftrightarrow \rho^*(v) = m$.

We first derive conditions for $v$ to have $\rho^{\psi}(v) = m$ and $\rho^*(v) = m$. 
We observe that $\rho^{\psi}(v) = m$ if and only if $v \in Y^{m+1} \setminus Y^m $ where
\begin{align}\label{eq:psi-IS} 
     &Y^{m+1} \,= \, (\,\neg \XX  \,\cup \,\Cpre(Y^m) \,)\,\cap\, \Cpre(Y^{m+1}) \,\cap \notag \\ 
    &\phantom{Y^{m+1}\,=\,\,}\,( \, \Lpre(Y^{m+1}) \cup \Preexists(Y^m) \cup \neg V^f\,),
    \\
    \vspace{0.2cm}
     &Y^{m} \quad = \,\, (\,\neg \XX \, \cup \, \Cpre(Y^{m-1}) )\, \cap \, \Cpre(Y^{m}) \, \cap \notag \\ 
     &\phantom{Y^{m+1}\,=\,\,}\,( \,\Lpre(Y^{m}) \, \cup \, \Preexists(Y^{m-1}) \,\cup\, \neg V^f\,).\notag
    \end{align}
    We know from the induction hypothesis that all nodes $ w \in Y^m$ (i.e., nodes with $\rho^{\psi}(w) \leq m-1$) have $\rho^{*}(w) \leq m-1$, and all nodes $w \in Y^m \setminus Y^{m-1}$ (i.e., nodes with $\rho^{\psi}(w) = m-1$) have $\rho^{*}(w) = m-1$.

    On the other hand, it follows from~\cref{eq:fairbuechi-pr,eq:fairBuechi-lifting} that $\rho^{*}(v) = m$ if and only if,
    \begin{alignat}{1}
    &\text{if $v \in \XX$,}\label{eq:lift-chi-IS}\\
    &\quad \text{if $v \in \Venv,$ min successor $w$ of $v$ has } \rho^*(w) = m-1,  \notag \\ 
    &\quad \text{if $v \in \Vsys,$ max successor $w$ of $v$ has } \rho^*(w) = m-1, \notag \\ 
    &\text{if $v \not \in \XX$,}  \label{eq:lift-nonchi-IS} \\
    &\quad \text{if } v \in V^f,
    \begin{cases} 
        \text{max fair successor $w \in E^f(v)$ of $v$ has} \notag \\\text{ \hspace{3,3cm} $\rho^*(w) = m$, or} \notag \\ 
        \text{min successor $w \in E(v)$ of $v$ has}\\ \text{ \hspace{3,3cm} $\rho^*(w) = m-1$,}
    \end{cases} \notag \\
    &\quad \text{if $v \in \Venv \setminus V^f ,$ min successor $w$ of $v$ has } \rho^*(w) = m,\notag \\ 
    &\quad \text{if $v \in \Vsys,$ max successor $w$ of $v$ has } \rho^*(w) = m. \notag
    \end{alignat}
A min/max successor of $v$ is a node $w \in E(v)$ that minimizes or maximizes the value $\rho^*(w)$. Similarly, a min/max fair successor is defined as $w \in E^f(v)$ with the minimum or maximum $\rho^*(w)$.

\bigskip
We now start the proof of the inductive step. 
We first prove the forward direction $\rho^{\psi}(v) = m \Rightarrow \rho^{*}(v) = m$. 

For $v \in \XX$, from~\cref{eq:psi-IS} we have $v \in \Cpre(Y^m) \setminus \Cpre(Y^{m-1})$. 
    
\begin{itemize}
\item If $v \in \Venv$, this gives that $v$ has a successor in $Y^m$, and it has no successors in $Y^{m-1}$. %
Thus, minimum $\rho^\psi$-successor $w$ of $v$ is in $Y^m \setminus Y^{m-1}$, and therefore has $\rho^\psi(w) = m-1$. From the induction hypothesis we obtain $\rho^{*}(w) = m-1$. 
\item If $v \in \Vsys$, all successors of $v$ are in $Y^m$, but not all of them are in $Y^{m-1}$. %
Therefore, the maximum $\rho^\psi$-successor $w$ of $v$ is in $Y^{m} \setminus Y^{m-1}$, and therefore has $\rho^{\psi}(w) = \rho^{*}(w) = m-1$ by the induction hypothesis.
\end{itemize}
Therefore, for $v \in \XX$,~\cref{eq:lift-chi-IS} is satisfied, i.e., we have $\rho^\psi(v) = m \Rightarrow \rho^{*}(v) = m$. 
For $v \in V^f \setminus \XX$, from~\cref{eq:psi-IS} we have $v$ is in \begin{equation}\label{eq:v-is-in} \Lpre(Y^{m+1}) \setminus \Lpre(Y^m)\,)\, \cup \, (\,\Preexists(Y^m) \setminus \Preexists(Y^{m-1}).\end{equation}
    If $v \in \, (\,\Preexists(Y^m) \setminus \Preexists(Y^{m-1})\,)$, then $v$ has a successor in $Y^m$, but no successors in $Y^{m-1}$. Thus, the minimum $\rho^{\psi}$-successor $w$ of $v$ has $\rho^\psi(w) = m-1$. The induction hypothesis implies $\rho^{*}(w) = m-1$.  
    Then, for $v \in (V^f \setminus \XX) \cap (\Preexists(Y^m) \setminus \Preexists(Y^{m-1}))$~\cref{eq:lift-nonchi-IS} is satisfied, i.e., we have $\rho^\psi(v) = m \Rightarrow \rho^{*}(v) = m$. 

    So far, we have shown for $v \in \XX$ and $v \in V^f \setminus \XX$ with a successor $w$ in $Y^m$ %
     that $\rho^\psi(v) = m \Rightarrow \rho^*(v) = m$ holds.
    We need to show the same for
    
    \noindent\textbf{\circled{1}} $v \in V \setminus (V^f \cup \XX)$, and
    
    \noindent\textbf{\circled{2}}
    $v \in V^f \setminus \XX$ with no successors in $Y^{m}$.
    
    Observe that for \circled{2}, as $v$ is in~\cref{eq:v-is-in}, all fair successors of $v$ are in $Y^{m+1}$. In fact, all of them are in $Y^{m+1}\setminus Y^{m}$ since $v$ has no successors in $Y^m$. For \circled{1}, due to~\cref{eq:psi-IS}, $v \in \, \Cpre(Y^{m+1}) \setminus \Cpre(Y^m)$. That is, for $v \in \Venv$, $v$ has a successor in $Y^{m+1}$, but no successors in $Y^m$. Similarly, for $v \in \Vsys$, all successors of $v$ are in $Y^{m+1}$, but not all of them are in $Y^m$.
    
     This allows us to construct a $\Venv$-subgraph $\subgraph = (V', \Vsys', \Venv', E')$ of $\mathcal{G}$ where $Y^{m+1}\setminus Y^m \subseteq V'\subseteq Y^{m+1}$. %

     \smallskip
     The nodes in \circled{1}\&\circled{2} constitute the non-dead-end nodes of $\subgraph$. Conversely, any node $v \in Y^{m+1} \setminus Y^m$ for which $\rho^*(v) = m$ is already established is a dead-end in $\subgraph$. That is, $v \in  (Y^{m+1}\setminus Y^m) \cap \XX$ and $ v \in (Y^{m+1}\setminus Y^m) \cap (V^f \setminus \XX)$ with a successor $w \in Y^m$ are dead-ends in $\subgraph$. Furthermore, any successor of non-dead-end nodes that has $\rho^*(w) < m$ (via the induction hypothesis for $w \in Y^m$), is dead-end in $\subgraph$ as well.

     For a node $v$ in \circled{1}$\&$\circled{2}, the successors of $v$ in $\subgraph$ are defined as follows:
      \begin{align*}
        &\text{for } \quad  v \in \Vsys, \quad &&E'(v) = E(v), \\
        &\text{for } \quad  v \in \Venv \setminus V^f, \quad &&E'(v) = \{w \in E(v) \mid w \in Y^{m+1}\},\\
        &\text{for }\quad  v \in \Venv \cap V^f, \quad &&E'(v) = E^f(v).
    \end{align*}
    Recall that for a node $v \in \Vsys$ in \circled{1}, all successors are in $Y^{m+1}$, %
    for a node $v \in \Venv$ in \circled{1}, it has a successor in $Y^{m+1}$ ($|E'(v)| \geq 1$), %
    and for a node $v \in V^f$ in \circled{2}, all fair successors are in $Y^{m+1} \setminus Y^m$. %
    Therefore, $\subgraph$ is a $\Venv$-subgraph. 
    Furthermore, $\subgraph$ contains all nodes in $Y^{m+1}\setminus Y^m$, all dead-end nodes $w$ in $\subgraph$ have $\rho^*(w) \leq m$, and all non-dead-end nodes in $\subgraph$ are in $V \setminus \XX$. 

    Observe that from the induction hypothesis we know that the nodes $v$ in \circled{1}$\&$\circled{2} have $\rho^*(v) \geq m$. Therefore, we only need to show that $\rho^* (v) \leq m$ also holds.    

    Apply the lifting function~\cref{eq:fairBuechi-lifting} to the nodes in $\subgraph$ presetting the dead-ends to their known $\rho^*\leq m$ values, and get a l.f.p. PM $\rho'$ on the subgraph. 
    Note that $\subgraph$ only restricts the outgoing edges of $\Venv$ nodes. Among \circled{1}\&\circled{2} nodes, the $\Venv\setminus V^f$ nodes get their value from their min successor and the $V^f$ nodes get their value from either their max fair successor or their min successor; restricting the outgoing edges of $\Venv \setminus V^f$ to a subset of the edges, and the outgoing edges of $V^f$ to fair ones that can increase the l.f.p. PM value. That is, $\rho'(v) \geq \rho^*(v)$ for all non-dead-end nodes in $\subgraph$, and for dead-end-nodes $\rho'(v) = \rho^*(v)$ as they are preset to $\rho^*(v)$ and not updated since they have no outgoing edges. 
    
    As there are no non-dead-end nodes that are in $\XX$, all non-dead-end nodes $v$ in $\subgraph$ will get $\rho'(v) \leq m$. This is because in the lifting algorithm, only the Büchi nodes increase the value of their successors. As we already know by the induction hypothesis that $\rho^*(v) \geq m$, this means $\rho^*(v) = m$ for all nodes $v \in Y^{m+1} \setminus Y^m$, i.e., all those with $\rho^\psi(v) = m$. 

    This completes the proof of the forward direction. Next we prove the backward direction $\rho^*(v) = m \Rightarrow \rho^\psi(v) = m$.

    \smallskip
    Take some $v$ with $\rho^*(v) = m$. First assume $v \in \XX$. Then by~\cref{eq:lift-chi-IS} and the induction hypothesis we know that (i) for $v \in \Venv$, $v$ has a minimum successor $w$ with $\rho^*(w) = \rho^\psi(w) = m-1$, and (ii) for $v \in \Vsys$, $v$ has a maximum successor $w$ with $\rho^*(w) = \rho^\psi(w) = m-1$. Therefore the successor $w$ in both cases satisfies $w \in Y^{m} \setminus Y^{m-1}$.
    This further entails that $v \in \Cpre(Y^m) \setminus \Cpre(Y^{m-1})$.

    Using this information, we will show that $v \in Y^{m+1} \setminus Y^{m}$. %
    According to~\cref{eq:psi-IS}, it is sufficient to show (I) $v \in (\neg \XX \cup \Cpre(Y^{m})) \setminus (\neg \XX \cup \Cpre(Y^{m-1}))$, (II) $v \in \Cpre(Y^{m+1})$, and (III) $v \in \Preexists(Y^m) \cup \neg V^f$.

    We showed that $v$ is in $\Cpre(Y^m) \setminus \Cpre(Y^{m-1})$. As $v \in \XX$, (I) follows directly.
    As $\Cpre(Y^m) \subseteq \Cpre(Y^{m+1})$, this immediately implies (II) as well. Finally (III) holds as for $ v \in \Vsys$, $\Vsys \subseteq \neg V^f$ and for $v \in \Venv$, $v \in \Preexists(Y^m)$ follows from $v \in \Cpre(Y^m)$.

    A similar argument works to get the same result for $v \in V^f \setminus \XX$ with a minimum successor $w$ with $\rho^*(w) = m-1$. From the induction hypothesis we have $\rho^\psi(w) = m-1$, i.e., $w \in Y^{m}\setminus Y^{m-1}$. By~\cref{eq:psi-IS}, in order to show $v \in Y^{m+1} \setminus Y^m$ it is sufficient to show (i) $v \in \Cpre(Y_{m+1})$ and (ii) $v \in \Preexists(Y_{m}) \setminus (\Lpre(Y_{m}) \cup \Preexists(Y_{m-1}))$. 

    (i) follows directly from $w$ being in $Y^m$ (and therefore in $Y^{m+1}$). (ii) $v \in \Preexists(Y_m) \setminus \Preexists(Y_{m-1})$ follows from $w$ being in $Y^m \setminus Y^{m-1}$, and $v \not \in \Lpre(Y_m)$ follows from $v$ not being in $Y^m$ ($\rho^\psi(v) \geq m$). 

    For $v \in \XX$ and for $v \in V^f \setminus \XX$ with a successor $w$ with $\rho^*(w) = m-1$, we have shown that $\rho^*(v) = m \Rightarrow \rho^\psi(v) = m$. Now we need to show the same for \\
    \noindent\textbf{\circled{1}}\quad $v \in V \setminus (V^f \cup \XX)$, and \\
    \noindent\textbf{\circled{2}}\quad $v \in V^f\setminus \XX$ with no successor $w$ with $\rho^*(w) \leq m-1$. 

    Similar to the other direction, we will construct a $\Venv$-subgraph $\subgr$ with the nodes $v$ with $\rho^*(v) = m$ and their successors in $\subgr$. Again, the nodes $w$ with confirmed $\rho^*(w) = m$ values and those with $\rho^*(w) = \rho^\psi(w) < m$ (by the induction hypothesis) that are successors of other nodes in $\subgr$ will be dead-ends. All other nodes with $\rho^*(w) = m$ (i.e., \circled{1}\&\circled{2}) will be the non-dead-end nodes of $\subgr$. 
    
     Analogously to the other direction of the proof, for a node $v$ in \circled{1}\&\circled{2} the successors of $v$ in $\subgr$ are as follows:
      \begin{align*}
        &\text{for } \,  v \in \Vsys,  &&E'(v) = E(v), \\
        &\text{for } \,  v \in \Venv \setminus V^f,  &&E'(v) = \{w \in E(v) \mid \rho^*(w) = m-1\},\\
        &\text{for }\,  v \in \Venv \cap V^f,  &&E'(v) = E^f(v).
    \end{align*}
    Because all $v$ in \circled{1}\&\circled{2} satisfy~\cref{eq:lift-nonchi-IS}, it can be observed that  for $v \in \Vsys$, all successors $w$ of $v$ have $\rho^*(w) = m$, for $v \in \Venv \setminus V^f$, there exists a successor $w$ of $v$ with $\rho^*(w) = m$, and for $v \in V^f$, all fair successors $w$ of $v$ have $\rho^*(w) = m$. Therefore, $\subgr$ is a $\Venv$-subgraph where all non-dead-end nodes are in $V \setminus \XX$. 

    As before, from the induction hypothesis we know that all nodes $v$ in \circled{1}\&\circled{2} have $\rho^\psi(v)\geq m$. Thus, we only need to show that $\rho^\psi(v) \leq m$ also holds. That is, we need to show that all non-dead-end nodes of $\subgr$ are in $Y^{m+1}$. 

    To show this, we write down $Y^{m+1}$ more explicitly than in~\cref{eq:psi-IS}, 
    \begin{align*}
     Y^{m+1} \,= \,\nu X.  & (\,\neg \XX  \,\cup \,\Cpre(Y^m) \,)\,\cap\, \Cpre(X) \,\cap \notag  \\ 
            & ( \, \Lpre(X) \cup \Preexists(Y^m) \cup \neg V^f\,).
    \end{align*}
    We already know that all the dead-end nodes in $\subgr$ are in $Y^{m+1}$. Since $X$ is a greatest fixpoint variable, it starts from $X^0 = V$ and saturates at $X^\infty = Y^{m+1}$.
    It suffices to show that all non-dead-end nodes in $\subgr$ are preserved during the iterations of the fixpoint variable $X$.

    A non-dead-end node $v$ in $\subgr$ can be lost during the iterations of $X$ if (1) it is not a fair node, and is not in $\Cpre(X^i)$ for some $X^i$, or (2) it is a fair node and is not in $\Cpre(X^i) \cap \Lpre(X^i)$ for some $X^i$. The dead-end nodes are all in $Y^{m+1}$, and therefore in $X^i$ for every $i$. In addition, for each $\Venv \setminus V^f$ node, there is a successor in $\subgr$, for each $\Vsys$ node, all the successors are in $\subgr$, and for each $V^f$ node, all the fair successors are in $\subgr$, and thus non of the non-dead-end nodes in $\subgr$ can be removed in any iteration $X^i$. 

    Therefore, all nodes in \circled{1}\&\circled{2} are in $Y^{m+1}$, and thus in $Y^{m+1} \setminus Y^m$. With this, we conclude that $\rho^*(v) = m \Rightarrow \rho^\psi(v) = m$, and therefore, the proof of~\cref{thm:fixpoint-pm-equivalence} is completed.
\end{proof}

%% file: appendix/app-experiments.tex
\section{Supplementary Material for Sec. 7}
\subsection{System Dynamics}
We evaluate our method using a 2D car model with state variables $(x, y)$. The control input is two-dimensional consisting of velocity and steering angle of the car. The input space is defined as
$$
\mathcal{U} = \{-0.2, -0.1, 0.1, 0.2\} \times \left\{0, \tfrac{\pi}{8}, \tfrac{\pi}{4}, \tfrac{3\pi}{8}, \tfrac{\pi}{2}, \tfrac{5\pi}{8},\tfrac{3\pi}{4},\tfrac{7\pi}{8},\pi\right\}.
$$
The (unknown) dynamics of the system are given by
$$
x_{k+1} = x_k + 10\delta v_k \cos(\theta_k) + w_1,
$$
$$
y_{k+1} = y_k + 10\delta v_k \sin(\theta_k) + w_2,
$$
where $\delta=1$ is the time discretization step, and $w_1$, $w_2$ represent additive stochastic noise. Noise support $\Omega$ is $[-1.5, 1.5] \times [-1.5, 1.5]$. Figure \ref{fig:mono} illustrates the area covered by the learned under- and over-approximations as a function of the number of samples in the dataset. In the experiments below, we focus on comparing incremental lifting with fixpoint recomputation.

\begin{figure}[ht]
    \centering
\includegraphics[width=0.8\linewidth]{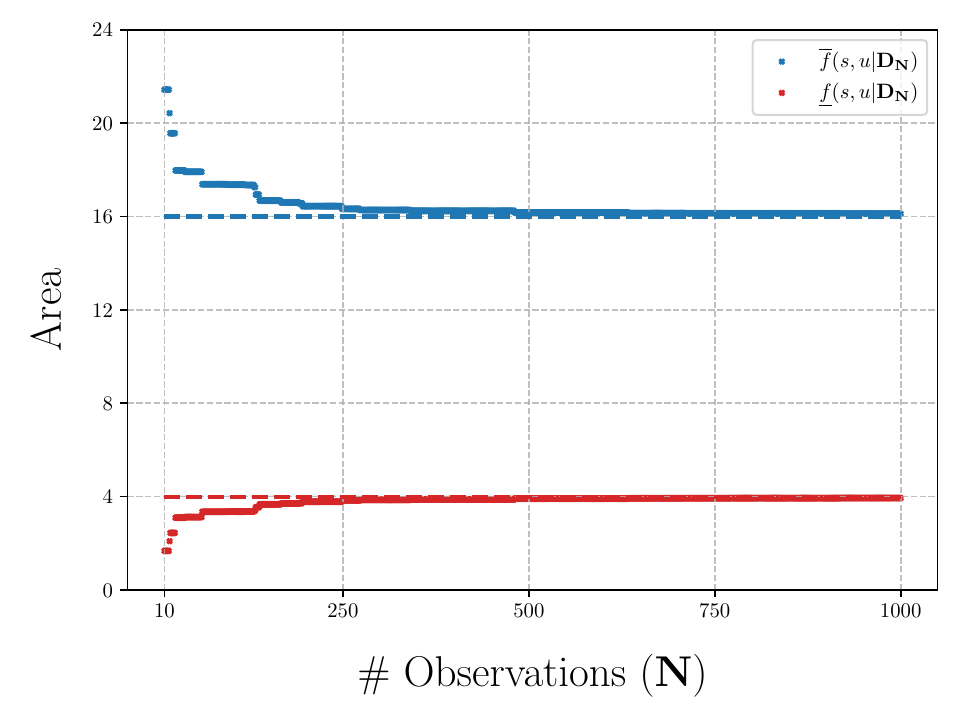}
\includegraphics[width=0.8\linewidth]{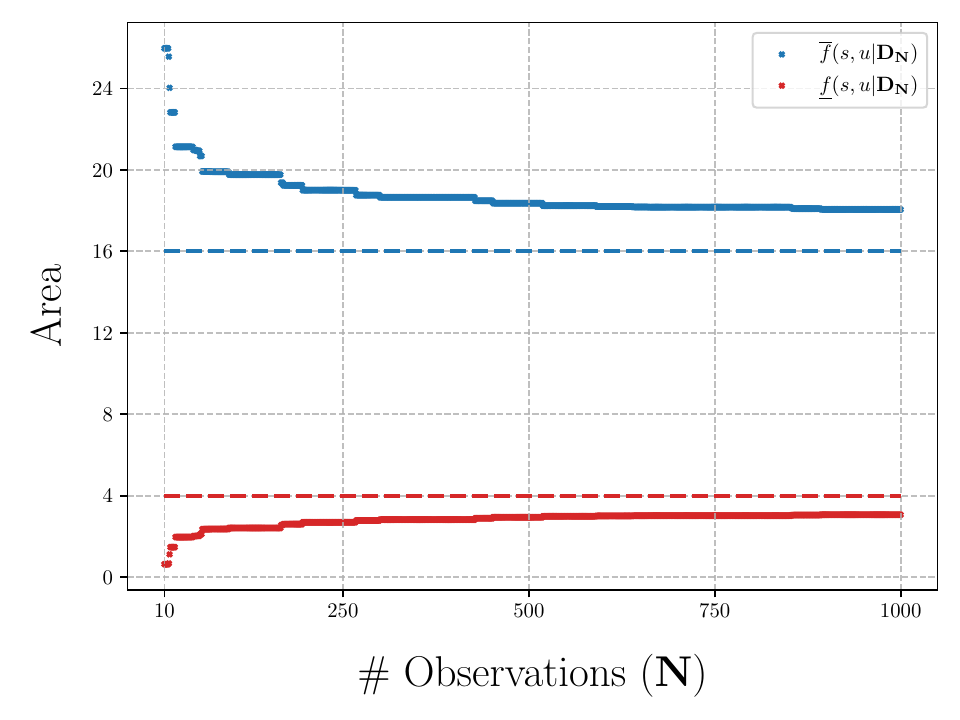}
    \caption{The area covered by the learned under- and over-approximations for a fixed state-input pair $(s, u)$ is shown as a function of the number of data points. Dashed lines represent the ground truth derived from the system dynamics. The top plot corresponds to $L_X = 1$, while the bottom plot uses a more conservative Lipschitz constant of $L_X = 1.25$. A larger (more conservative) Lipschitz constant leads to looser approximations.}
    \label{fig:mono}
\end{figure}

\subsection{Experiment 1}
In this experiment, we compare the runtime of our incremental lifting update against full recomputation when the game graph changes slightly. The robot must repeatedly reach a goal while avoiding a central obstacle. We assume limited data in one region---resulting in larger over- ($9 \times 9$ instead of $5 \times 5$) and smaller under-approximations (initially empty instead of $1 \times 1$). Figure \ref{fig:expadd1} shows the experiment on a $30 \times 30$ grid.
Initially, the winning domain excludes the obstacle and the low-data area. As new samples arrive, we refine the game graph and obtain the new winning domain both incrementally and by recomputation from scratch, and compare their runtimes across different room sizes. Our incremental method runs up to nearly $100$ times faster than full recomputation. See the main paper for detailed runtime plots over varying graph sizes.

\begin{figure}[t]
    \centering
\includegraphics[width=0.5\linewidth]{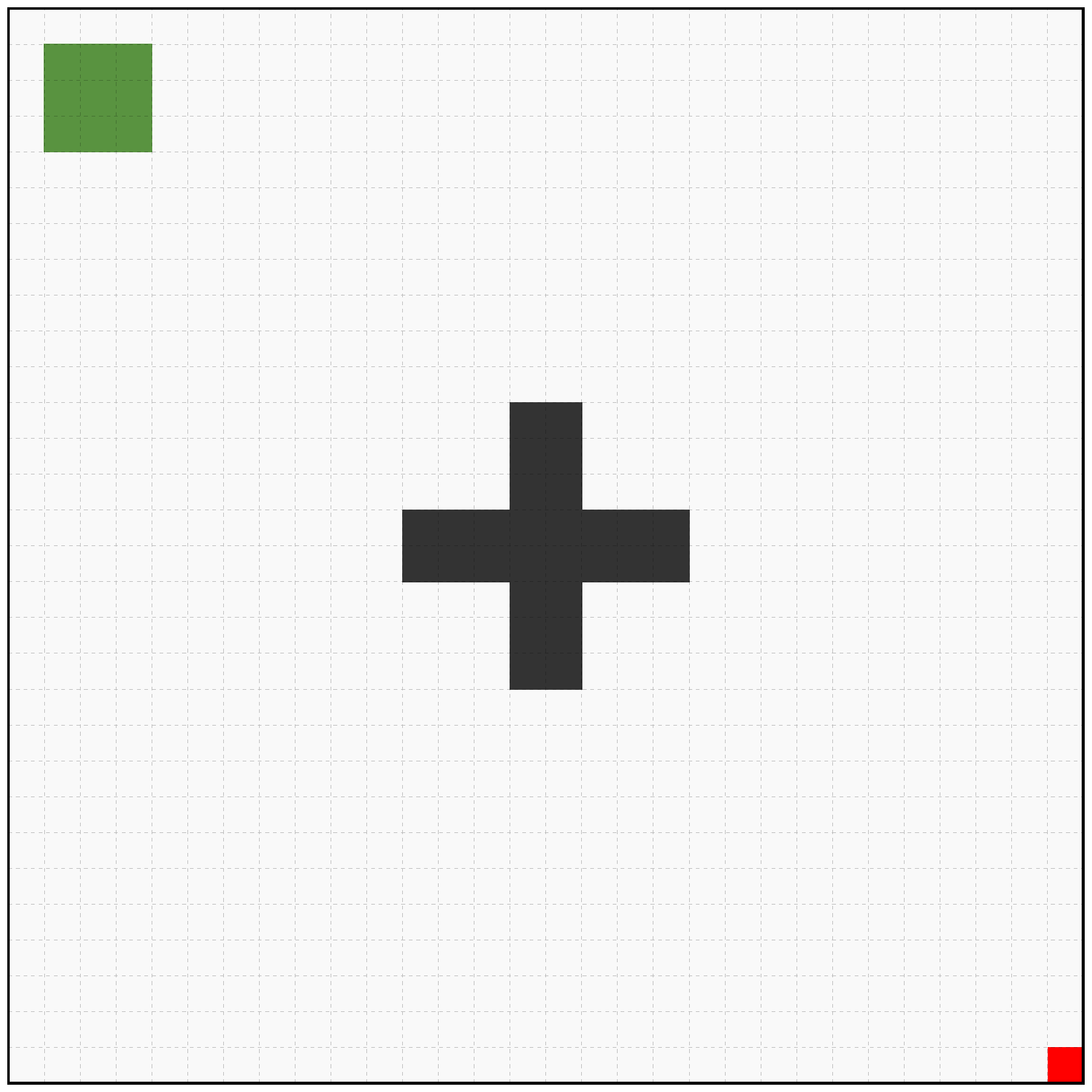}
    \caption{Green denotes the goal region; black denotes obstacles; red denotes the state with scarce data.}
    \label{fig:expadd1}
\end{figure}

\subsection{Experiment 2}
In this experiment, we extend the Experiment 1 by introducing  multiple randomly chosen low-data regions, producing larger perturbations to the game graph. Figure \ref{fig:exp2} shows the experiment on a $20 \times 20$ grid. As before, we record the time to update the winning domain once the game graph is modified. This setup highlights our incremental algorithm’s scalability and robustness under more substantial graph changes. Table~\ref{tab:exp2} summarizes the results for varying numbers of low-data regions and graph sizes.

\subsection{Experiment 3}
In this experiment, the robot operates in a five‐room apartment. Initially, observations are sparse inside all rooms. We then collect observations one room at a time and update the game graph. We recompute the winning domain both by complete recomputation and by our incremental lifting algorithm. Both approaches recover the same winning region, but incremental lifting runs substantially faster. This demonstrates the effectiveness of our incremental algorithm in iteratively solving games across multiple stages. See the main paper for detailed illustrations and runtime comparisons.

\begin{figure}[t]
    \centering
    \includegraphics[width=0.32\linewidth]{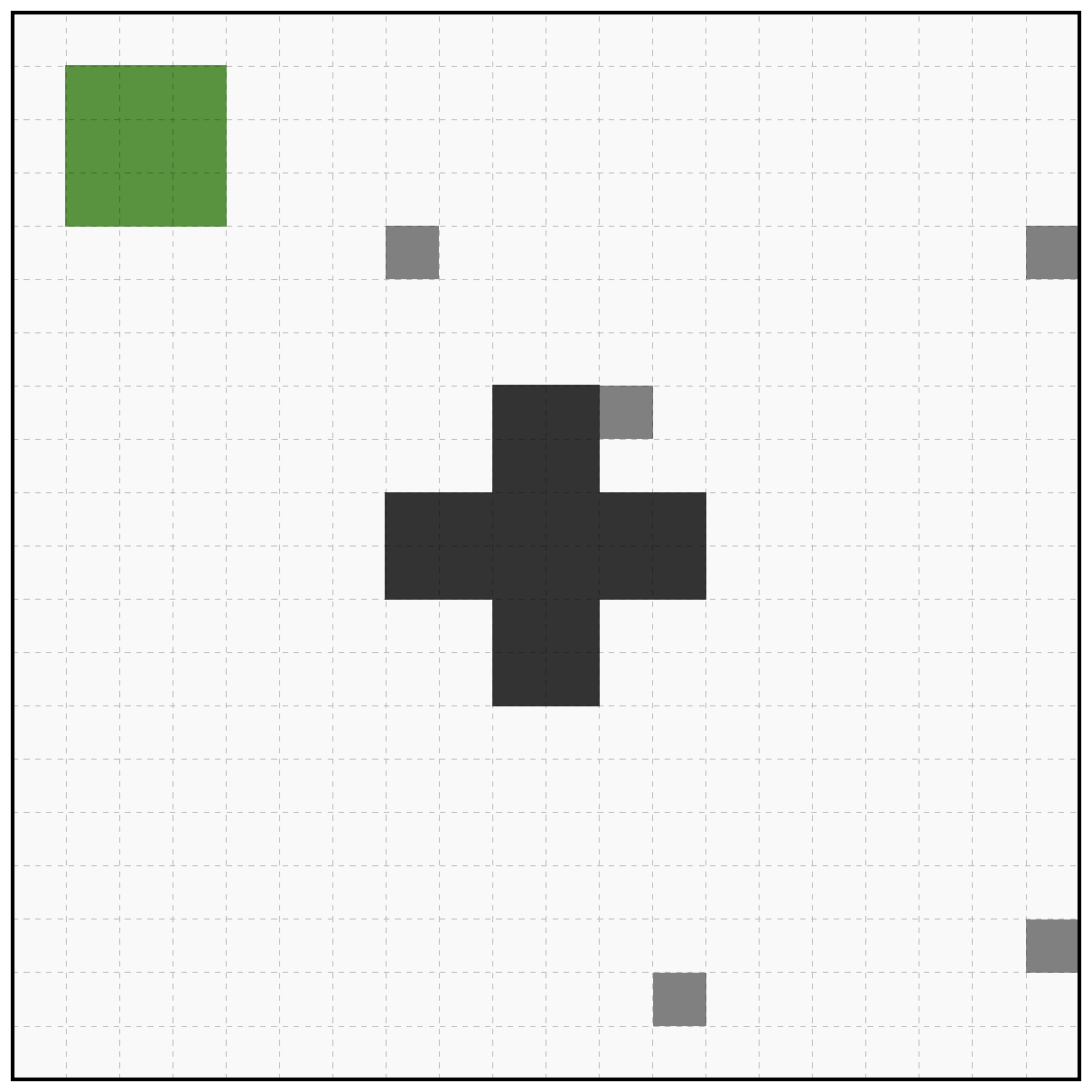}
    \includegraphics[width=0.32\linewidth]{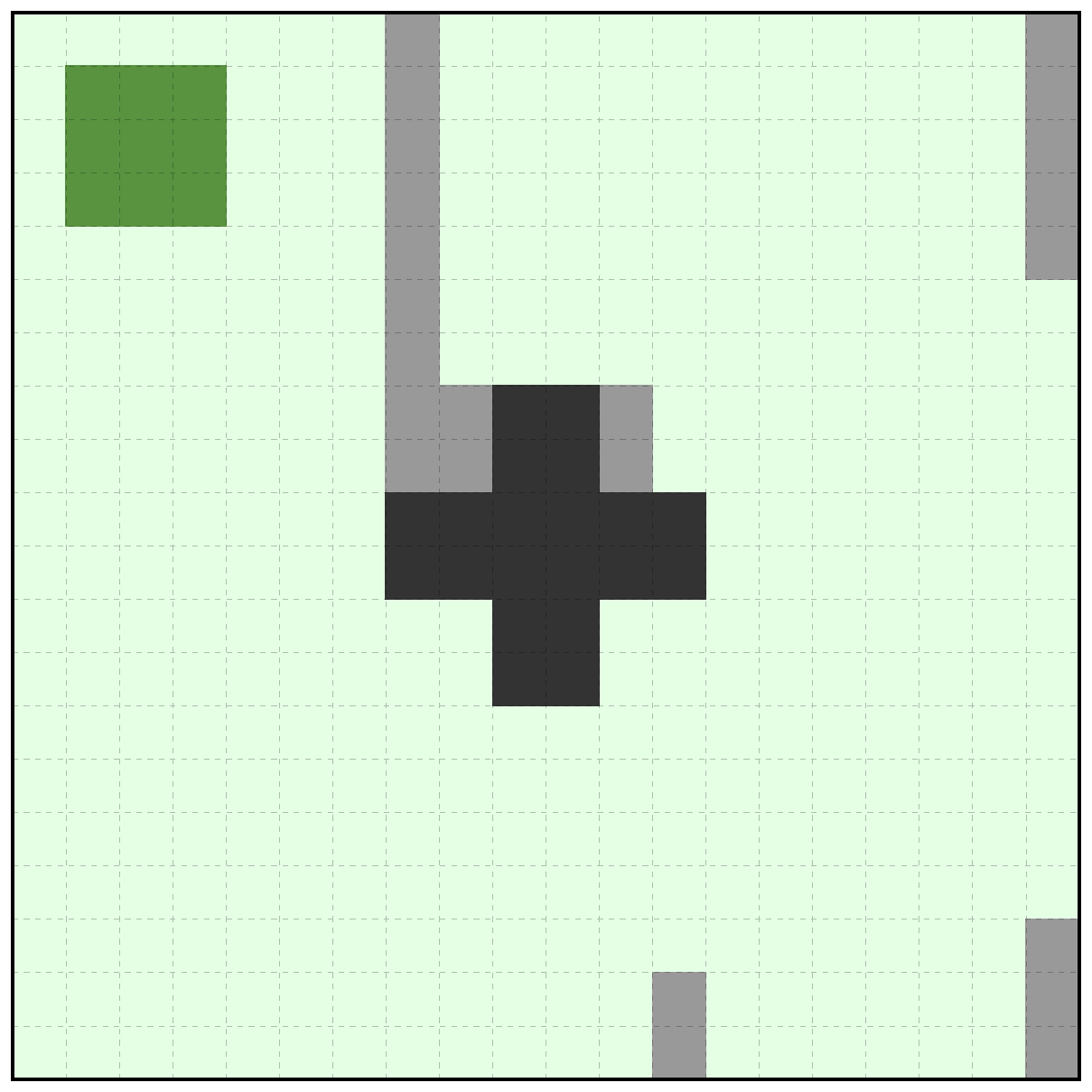}
    \includegraphics[width=0.32\linewidth]{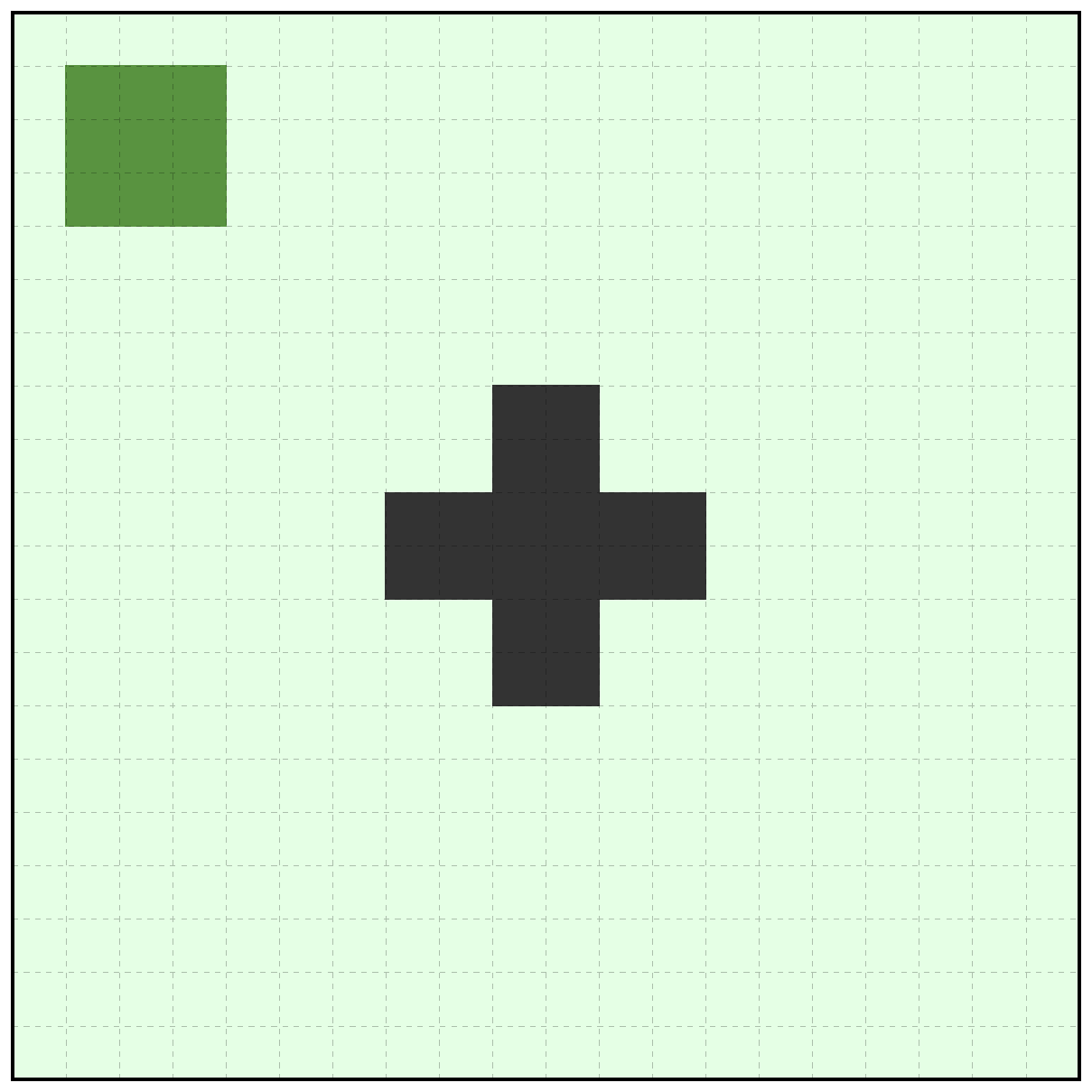}
    \caption{
Left: Gray indicate cells with scarce data (resulting in $9\times9$ over-approximations and empty under-approximations), whereas all other cells have $5\times5$ over-approximations and $1\times1$ under-approximations. 
Center: The initial winning domain (green). 
Right: The winning domain after incorporating more observations.
}
\label{fig:exp2}
\end{figure}

\begin{table*}[th]
\centering
\caption{Details of Experiment 2}
\begin{tabular}{ccccccc}
\toprule
\textbf{\#} & \textbf{Size} & \textbf{\#Low-data regions} & \textbf{\#States} & \textbf{Initialization (s)} & \textbf{ Incremental lifting (s)} & \textbf{Fixpoint recomputation (s)} \\
\midrule
1  & $25 \times 25$ & 5 & 328\,015 & 40.31 & 0.09  & 40.07  \\
2  & $25 \times 25$ & 10 & 329\,759 & 55.10 & 0.09 & 39.95  \\
3  & $25 \times 25$ & 15 & 332\,286 & 57.42 & 0.10 & 40.12  \\
4  & $25 \times 25$ & 20 & 333\,955 & 56.82 & 29.31 & 40.14 \\
5  & $25 \times 25$ & 25 & 337\,026 & 61.58 & 35.82 & 40.23  \\ \cmidrule{1-2}
6  & $35 \times 35$ & 10 & 697\,373 & 121.71 & 0.20 & 120.91  \\
7  & $35 \times 35$ & 20 & 700\,714 & 167.35 & 0.21 & 120.89 \\
8  & $35 \times 35$ & 30 & 707\,059 & 167.25 & 0.22 & 121.39  \\
9  & $35 \times 35$ & 40 & 709\,536 & 168.04 & 0.22 & 121.44  \\
10  & $35 \times 35$ & 50 & 715\,052 & 41.04 & 373.48 & 121.66 \\
\cmidrule{1-2}
11  & $45 \times 45$ & 15 & 1\,196\,297 & 257.89 & 0.34 & 254.37 \\
12  & $45 \times 45$ & 30 & 1\,204\,525 & 259.67 & 0.34 & 254.31 \\
13  & $45 \times 45$ & 45 & 1\,210\,111 & 353.86 & 0.36 & 255.87 \\
14  & $45 \times 45$ & 60 & 1\,216\,971 & 378.73 & 0.36 & 255.74 \\
15  & $45 \times 45$ & 75 & 1\,225\,921 & 384.49 & 0.42 & 256.35 \\
\bottomrule
\end{tabular}
\label{tab:exp2}
\end{table*}